\documentclass[11pt]{article}
\usepackage{latexsym}
\usepackage{amssymb, amsmath, amsthm}
\usepackage[top=1in,right=1in,left=1in,bottom=1in]{geometry}

\usepackage{float}
\usepackage{color,graphicx,graphics,epsfig,epstopdf}
\usepackage[dvipsnames]{xcolor}
\usepackage{cite,url}
\usepackage{dsfont}

\newtheorem{lem}{Lemma}[section]
\newtheorem{claim}{Claim}[section]
\newtheorem{observation}{Observation}[section]
\newtheorem{thm}{Theorem}
\newtheorem*{thmapp}{Theorem}
\newtheorem{thmappsec}{Theorem}[section]

\newtheorem{defn}{Definition}
\newtheorem{prob}{Problem}
\theoremstyle{remark}
\newtheorem{rem}{Remark}

\newcommand{\comment}[1]{}

\newcommand{\AutoAdjust}[3]{\mathchoice{ \left #1 #2  \right #3}{#1 #2 #3}{#1 #2 #3}{#1 #2 #3} }
\newcommand{\Xcomment}[1]{{}}

\newcommand{\InBrackets}[1]{\AutoAdjust{[}{#1}{]}}% {\left[{#1}\right]}
\newcommand{\Ex}[2][]{\operatorname{\mathbb E}_{#1}\InBrackets{#2}}

\newcommand{\Prx}[2][]{\operatorname{\mathbf{Pr}}_{#1}\InBrackets{#2}}
\def\prob{\Prx}

\newcommand{\infloor}[1]{\left\lfloor#1\right\rfloor}
\newcommand{\eqdef}{\stackrel{\textrm{def}}{=}}
\newcommand{\ind}[2][]{\mathds{1}_{\left[#1\right]}\left(#2\right)}

\newcommand{\vect}[1]{\ensuremath{\mathbf{#1}}}

\newcommand{\R}{\mathbb{R}}

\newcommand{\eps}{\varepsilon}

\newcommand{\be}{\begin{equation}}
\newcommand{\ee}{\end{equation}}
\newcommand{\bee}{\begin{equation*}}
\newcommand{\eee}{\end{equation*}}

\newcommand{\revenue}[0]{\mathrm{REV}}
\newcommand{\LW}{\mathrm{LW}}

\newcommand{\ncopy}{h}

\newcommand{\utili}[1][i]{u_{#1}}

\newcommand{\opt}{\mathrm{OPT}}

\newcommand{\val}{v}
\newcommand{\vali}[1][i]{\val_{#1}}
\newcommand{\vals}{\vect{\val}}
\newcommand{\valsmi}[1][i]{\vals_{\text{-}#1}}
\newcommand{\vij}[2][ij]{\val_{#1#2}}

\newcommand{\prices}{\vect{p}}

\newcommand{\pricejl}[2][]{{p_{#1}^{#2}}}

\newcommand{\rprices}{\overline{\vect{p}}}
\newcommand{\rpricei}[1][i]{{\overline{p}_{#1}}}
\newcommand{\rpricej}[1][j]{{\overline{p}_{#1}}}

\newcommand{\budget}{B}
\newcommand{\budgeti}[1][i]{B_{#1}}
\newcommand{\budgetsmi}[1][i]{\budget_{\text{-}#1}}

\newcommand{\bid}[2][]{b_{#1#2}}
\newcommand{\bidi}[1][i]{{b_{#1}}}
\newcommand{\bids}{\vect{b}}
\newcommand{\bidsmi}[1][i]{{\bids_{\text{-}#1}}}

\newcommand{\pbidi}[1][i]{b_{#1}'}

\newcommand{\abid}[2][]{b^{*}_{#1#2}}

\newcommand{\abidi}[1][i]{b_{#1}^*}

\newcommand{\alloc}{x}
\newcommand{\alloci}[1][i]{{\alloc_{#1}}}
\newcommand{\allocs}{\vect{x}}
\newcommand{\oalloc}{y}
\newcommand{\oalloci}[1][i]{{\oalloc_{#1}}}
\newcommand{\oallocs}{\vect{y}}
\newcommand{\varalloc}{z}

\newcommand{\strat}{s}
\newcommand{\strati}[1][i]{{\strat_{#1}}}
\newcommand{\strats}{\vect{\strat}}
\newcommand{\stratsmi}[1][i]{{\strats_{\text{-}#1}}}

\newcommand{\dist}{\mathcal{D}}
\newcommand{\disti}[1][i]{\dist_{#1}}
\newcommand{\distsmi}[1][i]{{\dists_{\text{-}#1}}}

\newcommand{\dists}{\vect{\dist}}

\newcommand{\pricedist}{\mathcal{F}}
\newcommand{\pricedists}{\vect{\pricedist}}

\newcommand{\good}{\mathcal{I}}
\newcommand{\notgood}{\overline{\mathcal{I}}}
\newcommand{\intapprox}[1]{\lfloor #1 \rfloor_{_\ncopy}}
\newcommand{\fracapprox}[1]{\left\{ #1 \right\}_{\ncopy}}
\newcommand{\sharej}[1][j]{#1^{\{\ell\}}}
\newcommand{\lpoa}{\textrm{LPoA}}

\begin{document}

\title{Liquid Price of Anarchy}
\author{
Yossi Azar\thanks{Tel Aviv University; \texttt{azar@tau.ac.il}, \texttt{mfeldman@tau.ac.il}, \texttt{alan.roytman@cs.tau.ac.il}}
\and
Michal Feldman\footnotemark[1]
\and
Nick Gravin\thanks{Microsoft Research; \texttt{ngravin@microsoft.com}}
\and
Alan Roytman\footnotemark[1]
}
\date{}

\maketitle

\thispagestyle{empty}

\begin{abstract}
% Version stoc1

Incorporating budget constraints into the analysis of auctions has become increasingly important, as they model practical settings more accurately.
The social welfare function, which is the standard measure of efficiency in auctions, is inadequate for settings with budgets, since there may be
a large disconnect between the value a bidder derives from obtaining an item and what can be liquidated from her.
The {\em Liquid Welfare} objective function has been suggested as a natural alternative for settings with budgets.
Simple auctions, like simultaneous item auctions, are evaluated by their performance at equilibrium using the Price of
Anarchy (PoA) measure -- the ratio of the objective function value of the optimal outcome to the worst equilibrium.
Accordingly, we evaluate the performance of simultaneous item auctions in budgeted settings by the
{\em Liquid Price of Anarchy} (LPoA) measure -- the ratio of the optimal Liquid Welfare to the Liquid Welfare obtained in the worst equilibrium.

Our main result is that the LPoA for mixed Nash equilibria is bounded by a constant when bidders are additive and items can be divided into sufficiently many discrete parts.
Our proofs are robust, and can be extended to achieve similar bounds for simultaneous second
price auctions as well as Bayesian Nash equilibria.
For pure Nash equilibria, we establish tight bounds on the LPoA for the larger class of fractionally-subadditive valuations.
To derive our results, we develop a new technique in which some bidders deviate (surprisingly) toward a non-optimal solution.  
In particular, this technique does not fit into the smoothness framework.

\end{abstract}

%\newpage
%\setcounter{page}{1}

\section{Introduction}
Budget constraints have become an important practical consideration in most existing
auctions, as reflected in recent literature (see, e.g., \cite{BCMX10,DLN08,BCGL12,S10}), because they model reality more accurately.
The issue of limited liquidity of buyers arises when transaction amounts
are large and may exhaust bidders' liquid assets, as is the case for privatization auctions in
Eastern Europe and FCC spectrum auctions in the U.S. (see, e.g., \cite{BK01}).
As another example, advertisers in Google Adword auctions are instructed to specify their budget even before specifying their bids and keywords.
Many other massive electronic marketplaces have a large number of participants with limited liquidity, which impose budget constraints.
Buyers would not borrow money from a bank to partake in multiple auctions on eBay, and even with available credit, they only have
a limited amount of attention, so that in aggregate they cannot spend too much money by participating in every auction online.
Finally, budget constraints also arise in small scale systems, such as the reality TV show Storage Wars, where people participate in cash-only auctions to win the content of an expired storage locker with an unknown asset.
%online market-place most of the bidders are budget constrained (not take a loan to bid on something)
%
%Traditional study - mechanism design with the .

%%,
%%%. This line of research focuses
%%%on designing incentive compatible mechanisms for agents that have budget constraints.
%%and many recent papers have incorporated budgets into the design of incentive compatible mechanisms,
%%including the works of Bhattacharya et al.~\cite{BCMX10}, Dobzinski et al.~\cite{DLN08}, Bei et %%al.~\cite{BCGL12}, and Singer~\cite{S10}.
%%% Practical and widely spread formats are simple and non truthful
%%% Simple auction formats: first- and second-price auctions are quite popular

Maximizing social welfare is a classic objective function that has been extensively studied within the context of resource allocation problems, and auctions in particular.
The \emph{social welfare} of an allocation is the sum of agents' valuations for their allocated bundles.
Unfortunately, in settings where agents have limited budgets (hereafter, {\em budgeted settings}), the social welfare objective fails to accurately capture what happens in practice.
Consider, for example, an auction in which there are two bidders and one item to be allocated
among the bidders.  One bidder has a high value but a very small budget, while the second bidder has a medium value along with a medium budget.
In this case, a high social welfare is achieved by allocating the item to the
bidder who values the item highly.
In contrast, most Internet advertising and electronic marketplaces (such as Google and eBay)
would allocate the item in the opposite way, namely to the bidder with a medium value and budget due to monetary constraints.
Hence, the social welfare objective is a poor model for how auctions are executed in reality.
Indeed, it seems reasonable to favor participants with substantial investments and engagement in the economical system to maintain a healthy economy.

%To this end, we study the efficiency of simple auctions according to an objective that takes budgets into account.
%Our analysis departs from the smoothness framework, which can only be applied in situations where the objective function dominates
%the sum of bidders' utilities (which social welfare satisfies but our objective does not).
%Hence, to overcome this issue, we develop new techniques which may be of independent interest in settings
%that do not fall into the smoothness framework.

%The standard measure of efficiency in auctions is {\em social welfare}, which is the sum of bidders' {\em values} for their obtained bundles.
Following Dobzinski and Paes Leme~\cite{DL14}, we measure the efficiency of outcomes in budgeted settings according to their {\em Liquid Welfare}
objective, motivated as follows.
In the absence of budgets, the value a buyer obtains from a given bundle captures their {\em willingness-to-pay} for the bundle.
According to this interpretation, the social welfare objective captures the maximum
revenue a seller can extract from buyers in a non-strategic setting.
In budgeted settings, however, the value a buyer receives from a bundle no longer captures how much revenue
can be extracted from them, since the revenue that can be extracted from a buyer is bounded by both the buyer's value and the buyer's budget.
To reconcile this discrepancy, Dobzinski and Paes Leme~\cite{DL14} proposed to evaluate the welfare of buyers according to their
{\em admissibility-to-pay}; that is, the minimum between the buyer's value for the allocated bundle and the buyer's budget. The
aggregate welfare according to this definition is termed the {\em Liquid Welfare} (LW). Indeed, the
Liquid Welfare is exactly the revenue a seller can extract from buyers in budgeted, non-strategic settings.

%As noted in~\cite{DL14}, there are a variety of reasons why the Liquid
%Welfare objective is the right measure of efficiency in settings with budgets.
%
%Another source of motivation for Liquid Welfare stems from the notion that the efficiency of a marketplace should only be measured
%by the funds available to players at the time of the auction, rather than any additional value players receive from winning items
%once the auction concludes.  This aspect is not modeled by the traditional social welfare objective.

%For a more concrete example of why Liquid Welfare might make more sense as a measure than social welfare, consider an auction in which there are $n$ buyers
%and $m\gg n$ items to be allocated among them.  Suppose each buyer values each item at $\char36 1$, and each buyer has a budget of $\char36 1$,
%except for a designated buyer who has a much larger budget $B \gg m$. In this auction, the social welfare is $m$ regardless of the chosen
%allocation, since buyers' budgets are not taken into account.
%In contrast, the Liquid Welfare objective favors outcomes that allocate more items to the designated buyer.
%Indeed, in real-life systems such as Internet advertising and electronic marketplaces, and especially in scenarios where
%goods have common values for participants, it seems reasonable to favor participants with substantial investments and
%engagement in the economical system.

In this work we study the efficiency of simple (non incentive compatible) auctions in budgeted settings.
The standard measure for quantifying the efficiency of simple auctions is the {\em Price of Anarchy} (PoA)~\cite{KP99,R09,ST13},
defined as the ratio of the optimal social welfare to the social welfare of the worst equilibrium.
In budgeted settings, it is thus natural to quantify the efficiency of simple auctions by the {\em Liquid Price of Anarchy} (LPoA),
defined as the ratio of the optimal Liquid Welfare to the Liquid Welfare of the worst equilibrium.

A prominent auction format, which has been extensively studied recently,
is the simultaneous item auction setting.
In such auctions, buyers submit bids simultaneously on all items, and the allocation and prices are determined separately for each
individual item, based only on the bids submitted for that item.
This format is similar to auctions used in practice (e.g., eBay auctions).
In the long line of works of Christodoulou et al.~\cite{CKS08},
Bhawalkar and Roughgarden~\cite{BR11}, Hassidim et al.~\cite{HKMN11}, Syrgkanis and Tardos~\cite{ST13}, and Feldman et al.~\cite{FFGL13},
it was shown that these simple auctions have nearly optimal social efficiency guarantees for a broad range of equilibrium concepts when
buyers have complement-free valuations but are not limited by budgets.

% \michal{summarize these 3 paragraphs in one paragraph, having the following points:
% 1. most work followed the smoothness framework, which ..... . 2. The smoothness framework is inadequate for liquid wefare because... , 3. We develop new techniques to quantify the PoA in budgetd settings, including a new type of hypothetical deviation and treating only a special set of carefully chosen bidders (see more details in the "our techinques" section below.). 4. This technique might prove useful in other settings that do not adhere to the smoothness framework.}

The most common framework for analyzing the Price of Anarchy of games and auctions is the {\em smoothness} framework (see, e.g., \cite{R09,ST13}).
Such techniques usually involve a thought experiment in which each player deviates toward some strategy related to the optimal solution, and hence
the total utility of all players can be bounded appropriately.
One important and necessary condition for applying the smoothness
framework is that the objective function dominate the sum of utilities (which holds for social welfare).
However, this technique falls short in the case of Liquid Welfare, since a bidder's utility can be
arbitrarily higher than their value, and in aggregate, bidders may achieve a total utility that is much larger than the Liquid Welfare at equilibrium.  
%To overcome this issue, we develop new techniques to bound the LPoA in budgeted settings,
%including a novel type of hypothetical deviation that is only done for a special set of carefully chosen bidders (see more details in
%the ``Our Techniques" section below).
To overcome this issue, we develop new techniques to bound the LPoA in budgeted settings.
Our techniques include a novel type of hypothetical deviation that is used to {\em upper bound} the aggregate utility of bidders (in addition to the traditional deviation that is used to lower bound it), and the consideration of a special set of carefully chosen bidders to engage in these hypothetical deviations (see more details in the ``Our Techniques" section below).
To the best of our knowledge, most prior techniques, including those that depart from the smoothness
framework (e.g., \cite{FFGL13}), examine the utility derived when \emph{every} player deviates toward the optimal solution.

With our new techniques at hand, we address the following question:
\emph{What is the Liquid Price of Anarchy of simultaneous item auctions in settings with budgets?}

\subsection*{Our Contributions}

%We quantify the efficiency loss of simultaneous item auctions in settings with budgets, using the {\em Liquid Price of Anarchy}
%measure.
We show that simultaneous item auctions achieve nearly optimal performance (i.e., a constant Liquid Price of Anarchy)
in many cases of interest.

Our main result concerns the case in which agent valuations are additive (i.e., agent $i$'s value for item $j$ is $\vij{}$ and the
value for a set of items is the sum of the individual valuations).

%We also assume , and every item $j$ has $\ncopy$ identical shares, where player $i$ values each share at $\frac{\vij{}}{\ncopy}$.

%\noindnet \textbf{Main result:} Simultaneous first price auctions achieve nearly optimal performance (i.e., a constant Liquid Price of Anarchy) for mixed and Bayes Nash Equilibria, with additive bidders and items that are discretely divisible.

\vspace{0.1in}

\noindent \textbf{Main theorem:} For simultaneous item auctions (first and second price) with additive bidders, the LPoA
with respect to mixed Nash equilibria and Bayesian Nash equilibria is constant. This assumes a divisible model where items can be divided into
sufficiently many discrete parts.

\vspace{0.1in}

\noindent We also show that for pure Nash equilibria, our results hold for more general settings.

\vspace{0.1in}

\noindent {\bf Theorem:}  For simultaneous item auctions (first and second price) with fractionally-subadditive bidders, the LPoA of pure Nash equilibria is 2, even in the indivisible model. This is tight.

\vspace{0.1in}

%\vspace*{-2pt}
%\begin{enumerate}
%\setlength{\itemsep}{1pt}
%\setlength{\parskip}{0pt}
%\setlength{\parsep}{0pt}
%%\item When players have additive valuations, we show that the
%%Liquid Price of Anarchy of mixed Nash equilibria and Bayesian Nash equilibria is $O\left(\frac{n}{\ncopy}\right)$.
%%Here, $n$ denotes the number of players and $\ncopy$ denotes the number of shares (copies) of each item.
%%In particular, the Liquid Price of Anarchy is constant for $\ncopy = \Omega(n)$.
%\item These results are tight in the following sense: the LPoA of mixed Nash equilibria is $\Omega\left(\frac{n}{\ncopy}\right)$, even with additive bidders.
%In particular, the LPoA is linear in $n$ when the number of shares is constant.
%\item For pure Nash equilibria, our results extend beyond additive bidders and hold even in the indivisible model (i.e., in the standard model, where $\ncopy = 1$). In particular, the LPoA with respect to pure Nash equilibria is 2 for fractionally-subadditive valuations, and this is tight.
%\end{enumerate}

\noindent The following remarks are in order:
\vspace*{-3pt}
\begin{enumerate}
%\michal{the following two points should be rewritten, shorter and clearer.}
\setlength{\itemsep}{1pt}
\setlength{\parskip}{0pt}
\setlength{\parsep}{0pt}
\item %We consider settings where bidders have {\em additive} valuations over the items.
In settings without budgets, simultaneous item auctions reduce to $m$ independent auctions (where $m$ is the number of items). In contrast, when agents have budget constraints, the separate auctions exhibit non-trivial dependencies even under additive valuations.
%many results typically extend to bidders with complex valuations (such as submodular, or fractionally-subadditive). In stark contrast, the budgeted setting already becomes challenging when bidders have additive valuations, due to budget constraints.
\item Our main result requires that every item can be divided into at least $\Omega(n)$ parts (where $n$ is the number of agents). If items can only be divided into a sublinear number of parts, then the LPoA is super constant.
\item Our LPoA result for pure Nash equilibria (in the indivisible model) holds for deterministic tie-breaking rules. Surprisingly, if the tie-breaking rule
is randomized, then the LPoA becomes linear in $n$ (even if agents play pure strategies and have additive valuations).
\end{enumerate}
\vspace*{-2pt}

\subsection*{Our Techniques}

% \michal{remember to mention ST'13 composition.}

% \michal{we can now shorten this a bit.}

The most common framework for analyzing the Price of Anarchy of games and auctions is the {\em smoothness} framework (see, e.g., \cite{R09,ST13}).
% The idea of smoothness is that at an equilibrium, each participant is guaranteed to obtain a certain level of utility that, in aggregate over all
% players, turns out to be closely related to the optimal solution.
% Such techniques usually involve a thought experiment
% in which each player deviates toward some baseline strategy related to the optimal solution, and hence the aggregate utility over
% all players can be bounded appropriately.
One important and necessary condition for applying the smoothness framework is that the objective
function dominate the sum of utilities, which clearly holds in the typical case of the social welfare objective,
but not the Liquid Welfare objective.
% However, this technique is inadequate in the case of the Liquid Welfare objective,
% since a bidder's utility can be
% arbitrarily higher than their value, and in aggregate,
% bidders may achieve a total utility that is much larger than the corresponding
% Liquid Welfare at equilibrium.

To overcome this obstacle, we introduce two new ideas. (1) In addition to deriving a lower bound
on the sum of bidders' utilities (following the traditional {\em deviations-towards-the-optimum} technique),
we also derive an {\em upper bound} on their utility as a function of the Liquid Welfare, using a novel {\em boosting deviation},
in which bidders bid more aggressively on items they receive in equilibrium. (2) Instead of summing the utility across all bidders,
we consider the utility derived from a carefully selected set of bidders.

Our analysis can be summarized by the following inequality:
let $\LW(\bids)$ denote the expected Liquid Welfare of a bid profile $\bids$, $\opt$ denote the optimal Liquid Welfare, and $\utili(\bids)$ denote the
expected utility of bidder $i$ under a bid profile $\bids$. For any equilibrium $\bids$, there exists a set of bidders $S$ and constants $c_1<1$ and $c_2>1$, such that
\begin{equation*}\label{eq:technique}
c_1 \cdot \opt \leq \sum_{i \in S}\utili(\bids) \leq c_2 \cdot \LW(\bids),
\end{equation*}
where the left inequality follows from the traditional {\em deviations-towards-the-optimum} technique,
and the right inequality follows from the new {\em boosting deviation} technique.

%\michal{I removed the paragraph about we are the first to consider a subset of deviators, since it already appears above.}
%To the best of our knowledge, most prior techniques, including those that depart from the smoothness
%framework (e.g., \cite{FFGL13}), examine the utility derived when \emph{every} player performs a {\em deviation-towards-the-optimum}.

Syrgkanis and Tardos~\cite{ST13} also addressed the PoA of simple auctions in settings with budgets.
They showed that the {\em social welfare} (SW) at equilibrium is at least a constant fraction of the optimal {\em Liquid Welfare}.
One might be tempted to leverage their results for bounding the LPoA. This approach, however, is inadequate, since the
LW at equilibrium can be arbitrarily smaller than the SW at equilibrium, even if items can be divided into arbitrarily
many parts (e.g., if all budgets are small and values are large)\footnote{Note that the SW at equilibrium can be arbitrarily larger
than the optimal LW, so the ratio studied by Syrgkanis and Tardos~\cite{ST13} can be either smaller or greater than 1. Our results imply that
whenever the LW at equilibrium is more than a constant factor smaller than the SW at equilibrium, it must be the case that the
optimal LW is also more than a constant factor smaller than the SW at equilibrium.}.
For this reason, the LPoA bounds we establish carry over to their setting, but not vice versa\footnote{Note, however, that in some sense our results and those of~\cite{ST13} are incomparable.
Although the bounds we establish on the LPoA imply bounds according to their PoA measure, they achieved results for more general equilibrium concepts and
valuation functions in some settings.}.
Since the focus of \cite{ST13} was to bound the SW at equilibrium (as opposed to LW), they established their bounds using the smoothness framework. In particular, they developed a powerful composition framework, in which they first obtained results for single-item auctions, and then showed that such auctions compose well to obtain more general results with any number of players and items. Note that it is not clear whether the composition framework is applicable in our setting.

\subsection*{Related Work}

There is a vast literature in algorithmic game theory that incorporates budgets into the design
of incentive compatible mechanisms. The paper of~\cite{BCMX10} showed that, in the case of one infinitely divisible good,
the adaptive clinching auction is incentive compatible under some assumptions.
Moreover, the work of~\cite{S10} initiated the design of incentive compatible mechanisms
in the context of reverse auctions, where the payments of the auctioneer cannot exceed a hard budget constraint
(follow-up works include~\cite{DPS11, BCGL12, CGL11, AGN14}).
%*Revenue
A great deal of work focused on designing incentive compatible mechanisms that approximately maximize the auctioneer's revenue
in various settings with budget-constrained bidders~\cite{LR96,MV08,PV08,BCIMS05,CMM11}.
Some works analyzed how budgets affect markets and non-truthful mechanisms~\cite{BK01,CG98}.

%*Social welfare
The earlier work~\cite{DLN08} on multi-unit auctions with budget constraints concerns the
design of incentive compatible mechanisms that always produce Pareto-optimal allocation. The results
in this line of work are mostly negative with a notable exception of mechanisms based on Ausubel's
adaptive clinching auction framework~\cite{A04}.

%The work of~\cite{DLN08} provided a negative result by showing that there are no incentive compatible mechanisms that always produce
%a Pareto-optimal allocation in the context of multi-unit auctions with budget constraints.
%They also provided a positive result based on Ausubel's adaptive clinching auction framework~\cite{A04},
%arguing that, in the case of two players, an arbitrary number of items, and public budgets, the adaptive
%clinching auction is the unique Pareto-optimal and incentive compatible auction.

% Pareto efficient outcomes in a truthful auction for public budgets -- Dobzinski, Lavi and Nisan [GEB2012]
% based on the Ausubel’s clinching auction framework [AmerEconReview1997]. Shown uniqueness of this
% truthful auction with Pareto-efficient solution.

%Truthful mechanism design with budgets
Some recent results concern the design of incentive compatible mechanisms with respect to the Liquid Welfare objective,
introduced by Dobzinski and Paes Leme~\cite{DL14}. They gave a constant approximation for the auction that sells a single divisible good to
additive buyers with budgets. In a follow-up work, Lu and Xiao~\cite{LX15} gave an $O(1)$-approximation for bidders with general valuations
in the single-item setting.

%In addition, they showed that, for any incentive
%compatible mechanism, it is impossible to approximate the Liquid Welfare objective to within a
%factor of $\frac{4}{3}$, and they gave an incentive compatible mechanism that matches this
%bound in a special case.  The authors also provided an $O(\log n)$-approximation when bidders have
%submodular valuation functions, and an $O(\log^2 n)$-approximation when bidders have
%subadditive valuations (for settings in which budgets are private).  There is also some follow-up
%work which improved these results~\cite{LX15}.

% A follow-up work~\cite{LX15} improved the $2$-approximation
% to a $\frac{1+\sqrt{5}}{2}$-approximation for public budgets and additive valuations.
% In the context of private budgets, they gave an $O(1)$-approximation, even when bidders can
% have any monotone, non-decreasing valuation function.

% *Liquid welfare: Incentive-compatible mechanisms. Single item, divisible good.
% Dobzinski and Paes Leme [ICALP2014]
% \cite{DL14}
% \cite{LX15}

%*Simultaneous auctions
A large body of literature is concerned with simultaneous item bidding auctions.
These simple auctions have been studied from a computational perspective, including the papers of Cai and Papadimitriou~\cite{CP14}
and Dobzinski et al.~\cite{DFK15}.  There is also extensive work addressing the Price of Anarchy of such simple
auctions (see~\cite{RT07} for more general Price of Anarchy results).  The work of Christodoulou et al.~\cite{CKS08}
initiated the study of simultaneous item auctions within the Price of Anarchy framework.
The authors showed that, for second price auctions, the social welfare of every Bayesian Nash equilibrium
is a $2$-approximation to the optimal social welfare, even for players with fractionally-subadditive
valuation functions.  A large amount of follow-up work~\cite{HKMN11,BR11,BR12,FFGL13,ST13,DHS13,CKST13}
made significant progress in our understanding of simultaneous item auctions, but none
of these works measures inefficiency with respect to the Liquid Welfare objective.

The two most closely related works to ours are those of Syrgkanis and Tardos~\cite{ST13}, along with
Caragiannis and Voudouris~\cite{CV14}, which both take the Liquid Welfare objective into account when measuring the inefficiency
of equilibria. The work of~\cite{ST13} gave a variety of Price of Anarchy results,
focusing on the development of a smoothness framework for broad solution concepts such as correlated equilibria and
Bayesian Nash equilibria, and exploring composition properties of various mechanisms. They extended their results to the
setting where players are budget-constrained, achieving similar approximation guarantees when comparing the {\em social
welfare achieved at equilibrium} to the {\em optimal Liquid Welfare}. In particular, their results imply an
$\frac{e}{e-1}$-approximation for simultaneous first price auctions, and a $2$-approximation for
all-pay auctions and simultaneous second price auctions under the no-overbidding assumption.
While~\cite{ST13} show that the social welfare at equilibrium cannot be much worse than the optimal Liquid Welfare,
one should note that the social welfare at equilibrium can be arbitrarily better
than the optimal Liquid Welfare (e.g., if all budgets are small, the optimal Liquid Welfare is small).
% One should note
% that the ratio between the social welfare at equilibrium and the optimal Liquid Welfare can be arbitrarily
% large (e.g., if all budgets are small, the optimal Liquid Welfare is small) or
% smaller than one (if all budgets are large, the optimal Liquid Welfare is the same as the optimal social welfare).
It is useful to note that, in general, the ratio between the Liquid Welfare at equilibrium and the social
welfare at equilibrium can be arbitrarily bad (if all budgets are small, then the Liquid Welfare of any
allocation is small, while players' values for received goods can be arbitrarily large).
% In contrast to~\cite{ST13},
% we study the two quantities which are bounded within constants, namely the {\em Liquid Welfare at equilibrium} to
% the {\em optimal Liquid Welfare}.

The paper of Caragiannis and Voudouris~\cite{CV14} also considered the scenario where players have budgets
and studied the same ratio we consider in this paper, namely the {\em Liquid Welfare at equilibrium} to the {\em optimal Liquid Welfare}.
They studied the proportional allocation mechanism, which concerns auctioning off one divisible
item proportionally according to the bids that players submit. They showed that, assuming players have
concave non-decreasing valuation functions,  the Liquid Welfare at coarse-correlated equilibria and Bayesian Nash equilibria
achieve at least a constant fraction of the optimal Liquid Welfare.
It should be noted that, for random allocations, they measure the benchmark at equilibrium ex-ante over the randomness of the
allocation, i.e., $\sum_{i=1}^n \min\{\Ex[\valsmi,\budgetsmi]{\vali(\alloci)},\budgeti\}$,
where $\vali$ is player $i$'s valuation, $\budgeti$ is player $i$'s budget, and $x_i$ denotes the allocation player $i$
receives. In contrast, for random allocations, we use the stronger ex-post measure of the expected Liquid Welfare at equilibrium given by
$\sum_{i=1}^n\Ex{\min\{\vali(\alloci),\budgeti\}}$.

\section{Model and Preliminaries}
\label{sec:prelim}
%\begin{itemize}
%\item number of bidders $n$
%\item number of items $m$
%\item bid: $\bid{}$ (macro \verb+\bid{}+)
%\begin{itemize}
%\item profile of bids: $\bids$ (macro \verb+\bids+)
%\item of buyer $i$ on item $j$: $\bid[i]{j}$ (macro \verb+\bid[i]{j}+)
%\item bids of buyer $i$: $\bidi$ (macro \verb+\bidi+ or possibly \verb+\bidi[k]+)
%\item bids of all but bidder $i$: $\bidsmi$ (macro \verb+\bidsmi+ or \verb+\bidsmi[k]+)
%\end{itemize}
%\item Budget of buyer $i$: $\budgeti$ (macro \verb+\budgeti+ or \verb+\budgeti[k]+)
%\item \nick{suggestion for} prices: $\price$ (macro \verb+\price+)
%\begin{itemize}
%\item each price per item: $\pricej$ (macro \verb+\pricej+ or \verb+\pricej[k]+)
%\item vector: $\prices$ (macro \verb+\prices+)
%\end{itemize}
%\item allocation $\alloc$ (macro \verb+\alloc+)
%\begin{itemize}
%\item for buyer $i$: $\alloci$ (macro \verb+\alloci+ or \verb+\alloci[k]+)
%\item profile: $\allocs$ (macro \verb+\allocs+)
%\end{itemize}
%\item valuation : $\val$ (macro \verb+\val+)
%\begin{itemize}
%\item of bidder $i$: $\vali$ (macro \verb+\vali+ or \verb+\vali[k]+)
%\item profile: $\vals$ (macro \verb+\vals+)
%\item all but $i$: $\valsmi$ (macro \verb+\valsmi+)
%\item for item $j$: $\vij{}$ (macro \verb+\vij[i]{j}+ or \verb+\vij{}+)
%\end{itemize}
%\end{itemize}

We consider {\em simultaneous item auctions}, in which $m$ heterogeneous items are sold to $n$ bidders (or players)
in $m$ independent auctions.  A bidder's \emph{strategy} is a bid vector $\bidi \in \R_{\ge 0}^m$, where
$\bid[i]{j}$ represents player $i$'s bid for item $j$.  We use $\bids$ to denote the bid profile
$\bids = (\bidi[1],\ldots,\bidi[n])$, and we will often use the notation $\bids = (\bidi, \bidsmi)$ to denote the
strategy profile where player $i$ bids $\bidi$ and the remaining players bid according to
$\bidsmi = (\bidi[1],\ldots,\bidi[i-1],\bidi[i+1],\ldots,\bidi[n])$.

The outcome of an auction consists of an allocation rule $\allocs$ and payment rule $\prices$.
The allocation rule $\allocs$ maps bid profiles to an allocation vector for each individual bidder $i$,
where $\alloci(\bids)=(\alloc_{i1},\dots,\alloc_{im})$ denotes the set of items won by player $i$.
In a \emph{simultaneous first price auction}, each item $j$ is allocated to the highest bidder
(breaking ties according to some rule) and the winner pays their bid. The total payment of bidder $i$ is
$p_i(\bids) = \sum_{j \in \alloci(\bids)}\bid[i]{j}.$

%
%the allocation rule $\allocs$ awards each item $j$ to the
%highest bidding player (breaking ties according to some rule),
%and the payment rule $\mathbf{p}$ charges each player their bid for each item that they win, namely
%$p_i(\bids) = \sum_{j \in \alloci(\bids)}\bid[i]{j}$
%(we sometimes just write $\alloci$ and $p_i$ when the context is clear).

Each player $i$ has a \emph{valuation function} $\vali$, which maps sets of items to $\mathbb{R}_{\geq 0}$ ($\vali$
captures how much player $i$ values item bundles), and a budget $\budgeti$. We assume that all valuations are
normalized and monotone, i.e., $\vali(\emptyset) = 0$ and $\vali(S) \le \vali(T)$ for any $i\in[n]$ and $S \subseteq T \subseteq [m]$.
We mostly consider bidders with additive valuations, i.e., $\vali(S)=\sum_{j\in S}\vij{}$ (where $\vij{}$ denotes agent $i$'s value for item $j$).
The \emph{utility} $\utili(\alloci(\bids))$ of each player $i$ is $\vali(\alloci(\bids))- p_i(\bids)=\sum_{j}\vij{}\cdot\alloc_{ij}- p_i(\bids)$
if $p_i(\bids) \le \budgeti$; and $\utili(\alloci(\bids))=-\infty$ if $p_i(\bids)>\budgeti.$
Buyers select their bids strategically in order to maximize utility.

\paragraph{Share model:}
For our results beyond pure Nash equilibria with deterministic tie-breaking rules, we focus on bidders with
additive valuations and consider a {\em share model}, in which
item $j$ is divided into $\ncopy$ identical shares
and player $i$ values each share at $\frac{\vij{}}{\ncopy}$.
% (here, $\vij{}$ denotes player $i$'s value for item $j$).
%\[
%\utili(\alloci(\bids))=
%\begin{cases}
%\sum_{j}\vij{}\cdot\alloc_{ij}- p_i(\bids)  & \quad \textrm{if } p_i(\bids) \leq \budgeti,\\
       %- \infty &  \quad\textrm{otherwise.}
%\end{cases}
%\]
%
%\[
%\utili(\alloci(\bids))=
%\begin{cases}
%\vali(\alloci(\bids))- p_i(\bids)  & \quad \textrm{if } p_i(\bids) \leq \budgeti,\\
       %- \infty &  \quad\textrm{otherwise}
%\end{cases}
%\]

%\[
%\utili(\alloci(\bids))=
%\begin{cases}\vali(\alloci(\bids))- p_i(\bids) =
%\sum_{j}\vij{}\cdot\alloc_{ij}- p_i(\bids)  & \quad \textrm{if } p_i(\bids) \leq \budgeti,\\
%       - \infty &  \quad\textrm{otherwise.}
%\end{cases}
%\]

\begin{defn}[Pure Nash Equilibrium]
A bid profile $\bids$ is a Pure Nash Equilibrium (PNE) if, for any player $i$ and any deviating bid $\pbidi$:
$\utili(\bidi,\bidsmi) \ge \utili(\pbidi,\bidsmi).$
\end{defn}
\noindent A mixed Nash equilibrium is defined similarly, except that bidding strategies can be randomized $\bidi\sim\strati$
and utility is measured in expectation over the joint bid distribution $\strats=\strati[1]\times\dots\times\strati[n]$.
\begin{defn}[Mixed Nash Equilibrium]
A bid profile $\strats$ is a mixed Nash equilibrium if, for any player $i$ and any deviating bid $\pbidi$:
$\Ex[\bids\sim\strats]{\utili(\bidi,\bidsmi)} \ge \Ex[\bidsmi\sim\stratsmi]{\utili(\pbidi,\bidsmi)}.$
\end{defn}
Note that, in general, we assume that the bidding space is discretized (i.e., each player
can only bid in multiples of a sufficiently small value $\eps$).  This is done to ensure that there
always exists a mixed Nash equilibrium, as otherwise we do not have a finite game.

We now give a definition about the welfare function we seek to optimize.
\begin{defn}[Liquid Welfare]
The Liquid Welfare, denoted by $\LW$, of an allocation $\allocs$ is given by $\LW(\allocs) = \sum_{i \in [n]} \min\{v_i(\alloci),\budgeti\}$.
For random allocations, we use the measure given by $\LW(\allocs) = \sum_{i \in [n]} \Ex{\min\{v_i(\alloci),\budgeti\}}$.
\end{defn}
For a given vector of valuations $\mathbf{v} = (v_1,\ldots,v_n)$, we use $\opt(\mathbf{v})$ to denote the value of the optimal outcome given by
$\opt(\mathbf{v}) = \max_{S_1,\ldots,S_n}\sum_i \min\{v_i(S_i),\budgeti\}$, where the sets $S_i$ form a partition of $[m]$ (i.e., $\cup_i S_i = [m]$
and $\forall i \neq j: S_i \cap S_j = \emptyset$).  We often use $\opt$ instead of $\opt(\mathbf{v})$ when the context is clear.
\begin{defn}[Liquid Price of Anarchy] Given a fixed valuation profile $\vals$,
the \emph{Liquid Price of Anarchy} (LPoA) is the worst-case ratio between the optimal Liquid Welfare and the expected Liquid Welfare at a Nash equilibrium (pure or mixed) and is given by
%Let $PNE$ be the set of all pure Nash equilibria (i.e., $PNE$ is a set of bid profiles where each profile is a pure Nash equilibrium).
%%in which players' bids satisfy no over-bidding and no over-budgeting.
%Then the Liquid Price of Anarchy (or LPoA in shorthand) is given by:
$$ \lpoa(\vals) = \sup_{\strats}\left\{\frac{\opt(\vals)}{\LW(\strats(\vals))} ~ \middle| ~\strats \in \text{Nash Equilibria} \right\}.$$
\end{defn}

\section{Main Result (Liquid Price of Anarchy)}
\label{sec:main}
%\section{Mixed Equilibria for First Price Auctions}

We consider a share model in which each item $j$ is divided into $\ncopy$ identical shares, where player $i$ values each
share at $\frac{\vij{}}{\ncopy}$ (here $\vij{}$ denotes player $i$'s value for item $j$).
For the sake of analysis,
we treat each share as a separate item, so that buyers can submit different bids on every single share.
A more realistic market clearing mechanism for one item would be one where
\vspace*{-2pt}
\begin{enumerate}
\setlength{\itemsep}{1pt}
\setlength{\parskip}{0pt}
\setlength{\parsep}{0pt}
\item Each buyer specifies how many shares they want to buy and which price they are willing to
      pay per share.
\item In decreasing order of bids and until the stock of shares lasts, each buyer receives
      their demanded number of shares while paying their bid per purchase.
\end{enumerate}
\vspace*{-2pt}
We note that our analysis carries over to this ``clearing house'' item auction with small adjustments, which we mention in Section~\ref{sec:extensions}
and discuss at the end of this section. 
%in Appendix~\ref{sec:appendix_misc}.

\begin{thm}
\label{thm:mixed_first_price}
The Liquid Price of Anarchy of simultaneous first price auctions is constant (at most $51.5$), when every item has at least 
$n$ equal shares (copies). 
If less shares $\ncopy$ are available then the LPoA is $O\left(\frac{n}{\ncopy}\right)$.
%The Liquid Price of Anarchy of simultaneous first price auctions, where every item
%has $\ncopy$ equal shares (copies), is $O\left(1+\frac{n}{\ncopy}\right)$ (at most $51.5$, when $\ncopy\ge n$).
\end{thm}

In what follows we build up notation and intuition toward the proof. Recall that
agents have additive valuations and submit bids on shares, and if they receive
an $\alloc_{ij}$ fraction of shares of item $j$, then their value is given by $v_i(\alloc_{ij}) = \vij{}\cdot \alloc_{ij}$. 
We further assume that the buyers bid according to a mixed Nash equilibrium $\bids\sim\strats$.
When buyers bid in simultaneous auctions, this essentially induces a distribution of prices over all shares of items
$\prices\sim\dist$ from a distribution $\dist$ (e.g., winning bids in first price auctions,
namely $\pricejl[j]{\ell} = \max_i \bid[i]{j}^\ell$ where $\bid[i]{j}^\ell$ is player $i$'s bid for share $\ell$ of item $j$).
In particular, for all items we
can define an ``expected price per item'' at equilibrium or just a ``price per item'' as
$\rprices=(\rpricei[1],\dots,\rpricei[m])$, where $\rpricej=\alpha\sum_{\ell=1}^{\ncopy} \Ex{\pricejl[j]{\ell}}$, for some $\alpha>1$
($\alpha=2$ will be sufficient for us). This induces a natural ``expected price per share,'' namely $\frac{\rpricej}{\ncopy}$.
One simple observation about $\rprices$ is the following:
\begin{observation}
Revenue is related to prices: $\revenue(\strats)=\frac{1}{\alpha}\sum_{j=1}^{m}\rpricej$, where $\revenue(\strats)$ denotes the expected
revenue at the equilibrium profile $\strats$.
\end{observation}
We next show that if players bid on some fraction of shares of item $j$ uniformly at random according to $\rpricej$,
then they win a large number of shares in expectation.  
%We prove the claim in Appendix~\ref{sec:appendix_misc}.
\begin{claim}\label{cl:copies}
For any item $j$, if a player bids on a $\delta$-fraction of shares chosen uniformly at random of item $j$ at a given price
$\frac{\rpricej}{\ncopy}$ per share, then the player receives in expectation at least $\ncopy\cdot\delta\cdot\left(1-\frac{1}{\alpha}\right)$ shares of the item
(i.e., at least a $\delta\cdot\left(1-\frac{1}{\alpha}\right)$-fraction of item $j$).
\end{claim}
\begin{proof}
Suppose towards a contradiction that the expected number of shares won by bidder $i$ is less than
$\delta \cdot \ncopy \cdot \left(1 - \frac{1}{\alpha}\right)$. In particular, it means that
\[\sum_{\ell=1}^{\ncopy}\Prx{i\text{ bids on share }\ell}\cdot \Prx{\pricejl[j]{\ell} < \frac{\rpricej}{\ncopy}}=
\sum_{\ell=1}^{\ncopy}\delta \cdot \Prx{\pricejl[j]{\ell} < \frac{\rpricej}{\ncopy}}
<\delta \cdot h \cdot \left(1 - \frac{1}{\alpha}\right).\]

We further use the definition of $\rpricej$ and Markov's inequality to obtain a contradiction as follows:
\begin{align*}
\frac{\rpricej}{\alpha} = \sum_{\ell=1}^\ncopy\Ex{\pricejl[j]{\ell}} &\ge
\sum_{\ell=1}^{\ncopy}\frac{\rpricej}{\ncopy} \cdot \Prx{\pricejl[j]{\ell} \ge \frac{\rpricej}{\ncopy}}
=\sum_{\ell=1}^{\ncopy}\frac{\rpricej}{\ncopy}\left(1 - \Prx{\pricejl[j]{\ell} < \frac{\rpricej}{\ncopy}}\right) >
\rpricej - \frac{\rpricej}{\ncopy} \cdot \ncopy \cdot\left(1 - \frac{1}{\alpha}\right).
\end{align*}
\end{proof}

When relating prices to Liquid Welfare we notice that
\begin{observation}
Revenue is bounded by the Liquid Welfare: $\revenue(\strats)\le \LW(\strats)$, where $\LW(\strats)$ denotes the
expected Liquid Welfare at the equilibrium profile $\strats$.
\end{observation}

We consider the following fractional relaxation of the allocation problem with the goal of optimizing Liquid Welfare.
\begin{alignat*}{5}
\text{Maximize}   \quad & \sum_{i=1}^n\sum_{j=1}^{m} \vij{}\cdot\varalloc_{ij}&  & \text{Liquid linear program (LLP)}\\
\text{Subject to} \quad & \sum_{j}\vij{}\cdot\varalloc_{ij} \le \budgeti \quad \forall i; \quad\quad
                           & \sum_{i}\varalloc_{ij} \le 1 \quad \forall j; \quad\quad
                           & \varalloc_{ij} \ge 0         \quad \forall  i,j.
\end{alignat*}

%\begin{defn} Liquid linear program (LLP)
%\begin{alignat*}{3}
% & \text{Maximize} & \sum_{i=1}^n\sum_{j=1}^{m} \vij{}\alloc_{ij}& \\
% & \text{Subject to} \quad& \sum_{j}\vij{}\alloc_{ij}& \leq \budgeti & \forall & i\\
%                 && \sum_{i}\alloc_{ij} &\leq 1 \quad & \forall & j\\
%                 && \alloc_{ij} &\geq 0 & \forall & i,j
%\end{alignat*}
%\end{defn}
We denote by $\oallocs=(\oalloc_{ij})$ the optimal solution to LLP. Notice that
the optimal fractional solution for the Liquid Welfare would never benefit from
allocating a set of items to a player such that their value for the set exceeds their budget.
The solution to LLP gives an upper bound on the optimal Liquid Welfare $\opt$.
\begin{observation}
The optimal fractional solution to LLP is better than the optimal allocation:
$\sum\limits_{i=1}^n\sum\limits_{j=1}^{m} \vij{}\cdot\oalloc_{ij}\ge \opt$.
\end{observation}

We now define some notation that will be useful in order to obtain our result.  We
let $q_{ij}$ be the expected fraction of shares that player $i$ receives from item $j$ at
an equilibrium strategy $\strats$.
In addition, for each agent $i$, we consider a set of high value items $J_i \eqdef \left\{j \mid \vij{} \ge \rpricej\right\}$.
We further define $Q_i$ to be the probability that $\vali(\alloci) \ge \budgeti$ at equilibrium (recall that $\alloci$ denotes the random set
that player $i$ receives in the mixed Nash equilibrium).
We also define three sets of bidders, the first two of which are for budget feasibility reasons and the last of which is for bidders that often
fall under their budget in equilibrium (these sets need not be disjoint).  In particular, for a fixed parameter $\gamma >1$
to be determined later, we define sets $\good_1$, $\good_2$, and $\good_3$:
\[
\good_1\eqdef\left . \Big\{i~~\middle| \right .  \gamma\sum_{j\in J_i} \rpricej\cdot q_{ij}\le\budgeti\Big\}, \quad
\good_2\eqdef\left . \Big\{ i~~\middle|\right . \sum_{j\in[m]}\frac{\rpricej}{\ncopy}\le\budgeti\Big\}, \quad\text{and}\quad
\good_3\eqdef\left . \Big\{ i~~\middle|\right . Q_i \leq \frac{1}{2 \gamma}\Big\}.
\]
Throughout our proof, we focus on bidders in the set $\good \eqdef \good_1 \cap \good_2 \cap \good_3$. We define sets $\notgood_1\eqdef[n]\setminus\good_1$,
$\notgood_2\eqdef[n]\setminus\good_2$, $\notgood_3\eqdef[n]\setminus\good_3$, and $\notgood\eqdef[n]\setminus\good$. To this end, we need
to argue that bidders outside of the set $\good$ do not contribute a lot to the Liquid Welfare at equilibrium $\strats$.  
%We prove the claim in Appendix~\ref{sec:appendix_misc}.

\begin{claim}
\label{cl:notgood_budgets}
The total budget of players in $\notgood$ is small:
$\sum_{i \in \notgood}\budgeti < \alpha \cdot \left(\gamma+\frac{n}{\ncopy}\right) \cdot \revenue(\strats) + \sum_{i \in \notgood_3}\budgeti$.
\end{claim}
\begin{proof}
We first consider agents in $\notgood_1$ and obtain
$$
\sum_{i \in \notgood_1}\budgeti < \gamma\sum_{i \in \notgood_1}\sum_{j\in J_i} \rpricej\cdot q_{ij} \leq \gamma\sum_i\sum_{j \in J_i}\rpricej\cdot q_{ij}
\leq \gamma \sum_j \rpricej\sum_{i:j \in J_i}q_{ij} \leq \gamma \sum_j \rpricej \leq \gamma \cdot \alpha \cdot \revenue(\strats),
$$
where the second to last inequality follows from the fact that $\sum_i q_{ij} \leq 1$ since we have a valid fractional allocation.
Next, we consider agents in $\notgood_2$ and obtain
$$
\sum_{i \in \notgood_2}\budgeti < \sum_{i \in \notgood_2}\sum_{j}\frac{\rpricej}{\ncopy} \leq \frac{n}{\ncopy}\sum_j\rpricej \leq \frac{n}{\ncopy} \cdot \alpha \cdot \revenue(\strats).
$$
Combining these bounds for agents in $\notgood_1$ and $\notgood_2$ we have
\bee
\sum_{i \in \notgood}\budgeti \leq \sum_{i \in \notgood_1}\budgeti + \sum_{i \in \notgood_2}\budgeti + \sum_{i \in \notgood_3}\budgeti
\leq \alpha \cdot \left(\gamma+\frac{n}{\ncopy}\right) \cdot \revenue(\strats) + \sum_{i \in \notgood_3}\budgeti. \qedhere
\eee
\end{proof}

To achieve our result, we essentially consider two main ideas for player deviations in set $\good$ (each idea
actually consists of two parts). The first idea is to use the solution to LLP as guidance to claim that players can extract a large amount of
value relative to the optimal solution.  However, for players to deviate, we must first round the fractional solution $\oallocs$
into an integral solution (here, by integral, we mean a multiple of $\frac{1}{\ncopy}$ since this represents the fraction
of shares a player receives out of $\ncopy$ copies).  Define the first LLP deviation (integral part) to be $\bids_1^{\infloor{i}}=(\pbidi,\bidsmi)$,
where in $\pbidi$ buyer $i$ bids on a random $\intapprox{\oalloc_{ij}}$-fraction of each item $j\in J_i$ with price $\rpricej$
(here, $\intapprox{\oalloc}=\frac{\lfloor\oalloc\cdot\ncopy\rfloor}{\ncopy}$ and the bid per share is $\frac{\rpricej}{\ncopy}$).  Define the second LLP deviation (fractional part) to be
$\bids_1^{\{i\}}=(\pbidi,\bidsmi)$, where in $\pbidi$ buyer $i$ bids on a random $\fracapprox{\oalloc_{ij}}$-fraction
of each item $j\in J_i$ with price $\rpricej$ (here, $\fracapprox{\oalloc}=\frac{1}{\ncopy}$ if $\oalloc>0$, and
$\fracapprox{\oalloc}=0$ otherwise).  We note that both LLP deviations $\bids_1^{\infloor{i}}$ and $\bids_1^{\{i\}}$ are feasible, 
since $\vij{} \ge \rpricej$ for every $j\in J_i$, and $\sum_{j}\vij{}\cdot\oalloc_{ij} \le \budgeti$ as $\oallocs$ is a solution to LLP
along with $\good \subseteq \good_2$.  Moreover, for any $\oalloc_{ij}$, we have $\intapprox{\oalloc_{ij}} + \fracapprox{\oalloc_{ij}} \geq \oalloc_{ij}$.

%as long as $\ncopy\ge n$ for every $i\in\good$, since by definition $\good \subseteq \good_2$.  Moreover, for any $\oalloc_{ij}$,
%we have $\intapprox{\oalloc_{ij}} + \fracapprox{\oalloc_{ij}} \geq \oalloc_{ij}$.

\begin{lem}[LLP deviations]
\label{lem:first_deviation}
Buyers in $\good$ at equilibrium $\strats$ derive large value:
\begin{align*}
\sum_{i \in \good}\sum_{j}\vij{}\cdot q_{ij} &\ge \left(\frac{1}{2}-\frac{1}{2\alpha}\right)
\left(\opt - \alpha\left(1+\gamma+\frac{n}{\ncopy}\right)\revenue(\strats) - \sum_{i \in \notgood_3}\budgeti \right).
\end{align*}
\end{lem}
\begin{proof} For the integral part of the LLP deviation, since $\Ex[\bids\sim\strats]{\vali(\bids)}\ge\Ex[\bids\sim\strats]{\utili(\bids)}$ and
$\strats$ is a mixed Nash equilibrium, we have:
\begin{align}
\label{eq:nash_first_deviation_integral}
\sum_{i\in\good}\sum_{j}\vij{}\cdot q_{ij} &\ge \sum_{i\in\good}\Ex[\bids\sim\strats]{\utili(\bids)}
\ge \sum_{i \in\good} \Ex[\bidsmi\sim\stratsmi]{\utili\left(\bids_1^{\infloor{i}}\right)}
\nonumber\\
&\ge \sum_{i \in\good}\sum_{j \in J_i}\left(1-\frac{1}{\alpha}\right)\cdot\intapprox{\oalloc_{ij}} \left(\vij{} - \rpricej\right),
\end{align}
where to derive the last inequality we use Claim~\ref{cl:copies}. Similarly, for the fractional part of the LLP deviation we have:
\begin{align}
\label{eq:nash_first_deviation_fractional}
\sum_{i\in\good}\sum_{j}\vij{}\cdot q_{ij} &\ge \sum_{i\in\good}\Ex[\bids\sim\strats]{\utili(\bids)}
\ge \sum_{i \in\good} \Ex[\bidsmi\sim\stratsmi]{\utili\left(\bids_1^{\{i\}}\right)}
\nonumber\\
&\ge \sum_{i \in\good}\sum_{j \in J_i}\left(1-\frac{1}{\alpha}\right)\cdot\fracapprox{\oalloc_{ij}} \left(\vij{} - \rpricej\right).
\end{align}

Combining Equation~\eqref{eq:nash_first_deviation_integral} and Equation~\eqref{eq:nash_first_deviation_fractional} we get
%\begin{align}
%\label{eq:nash_first_deviation}
%2\sum_{i\in\good}\sum_{j}\vij{}\cdot q_{ij} &\ge
%\sum_{i \in\good}\sum_{j \in J_i}\left(1-\frac{1}{\alpha}\right)\cdot\left(\intapprox{\oalloc_{ij}}+\fracapprox{\oalloc_{ij}}\right) \left(\vij{} - \rpricej\right)
%\\
%&\ge
%\sum_{i \in\good}\sum_{j \in J_i}\left(1-\frac{1}{\alpha}\right)\cdot \oalloc_{ij}\left(\vij{} - \rpricej\right)
%\geq
%\left(\frac{1}{2}-\frac{1}{2\alpha}\right)\sum_{i \in\good}\sum_{j \in J_i} \vij{}\cdot\oalloc_{ij}, \nonumber
%\end{align}
%where the last inequality follows from the fact that $\vij{} \geq 2\rpricej$ for each $j \in J_i$.
\begin{align}
\label{eq:nash_first_deviation}
2\sum_{i\in\good}&\sum_{j}\vij{}\cdot q_{ij} \ge
\sum_{i \in\good}\sum_{j \in J_i}\left(1-\frac{1}{\alpha}\right)\cdot(\left\intapprox{\oalloc_{ij}}+\fracapprox{\oalloc_{ij}}\right) \left(\vij{} - \rpricej\right)
\nonumber\\
&\ge
\sum_{i \in\good}\sum_{j \in J_i}\left(1-\frac{1}{\alpha}\right)\cdot \oalloc_{ij}\left(\vij{} - \rpricej\right)
=
\left(1-\frac{1}{\alpha}\right)\left(\sum_{i \in\good}\sum_{j \in J_i} \vij{}\cdot\oalloc_{ij} -
\sum_{i \in\good}\sum_{j \in J_i}\rpricej\cdot\oalloc_{ij}\right).
\end{align}
We further estimate
\begin{align*}
\sum_{i \in \good}\sum_{j\in J_i}\vij{}\cdot\oalloc_{ij} &=
\sum_{i,j}\vij{}\cdot\oalloc_{ij} - \sum_{i\in\notgood}\sum_{j}\vij{}\cdot\oalloc_{ij} - \sum_{i\in\good}\sum_{j \not\in J_i}\vij{}\cdot\oalloc_{ij}\\
&\ge \sum_{i,j}\vij{}\cdot\oalloc_{ij} - \sum_{i\in\notgood}\budgeti - \sum_{i \in \good}\sum_{j \not\in J_i}\rpricej\cdot \oalloc_{ij},
\end{align*}
where in the last inequality we used the condition from the LLP that $\sum_{j}\vij{}\cdot\oalloc_{ij}\le\budgeti$ and that $\vij{}\le\rpricej$ for each $j\not\in J_i$.
We substitute the last estimate into Equation~\eqref{eq:nash_first_deviation} and get
%\begin{align*}
%\sum_{i \in \good}\sum_{j\in J_i}\vij{}\cdot\oalloc_{ij} &=
%\sum_{i,j}\vij{}\cdot\oalloc_{ij} - \sum_{i\not\in\good}\sum_{j}\vij{}\cdot\oalloc_{ij} - \sum_{i\in\good}\sum_{j \not\in J_i}\vij{}\cdot\oalloc_{ij}\\
%&\ge \sum_{i,j}\vij{}\cdot\oalloc_{ij} - \sum_{i\not\in\good}\budgeti - 2\sum_{i \in \good}\sum_{j \not\in J_i}\rpricej\cdot \oalloc_{ij},
%\end{align*}
%where in the last inequality we use the condition from the LP that $\sum_{j}\vij{}\cdot\oalloc_{ij}\le\budgeti$ for each $i$ and that $\vij{} \le 2\rpricej$ for each $j\not\in J_i$.
%We substitute the last estimate into Equation~\eqref{eq:nash_first_deviation} and get
\begin{align*}
2\sum_{i\in\good}\sum_{j}\vij{}\cdot q_{ij} &\ge
\left(1-\frac{1}{\alpha}\right)\left(
\sum_{i,j}\vij{}\cdot\oalloc_{ij} - \sum_{i\in\notgood}\budgeti - \sum_{i \in \good}\sum_{j \not\in J_i}\rpricej\cdot \oalloc_{ij} -
\sum_{i \in\good}\sum_{j \in J_i}\rpricej\cdot\oalloc_{ij}\right)\\
&=
\left(1-\frac{1}{\alpha}\right)\left(
\sum_{i,j}\vij{}\cdot\oalloc_{ij} - \sum_{i \in \good}\sum_{j}\rpricej\cdot \oalloc_{ij} - \sum_{i\in\notgood}\budgeti\right)\\
&\ge
\left(1-\frac{1}{\alpha}\right)\left(\opt - \alpha\left(1+\gamma+\frac{n}{\ncopy}\right)\revenue(\strats) - \sum_{i \in \notgood_3}\budgeti\right),
\end{align*}
where the last inequality follows from the LLP constraint that $\sum_i \oalloc_{ij} \leq 1$ for each $j$, the observation
$\sum_{j}\rpricej= \alpha\cdot\revenue(\strats)$, and Claim~\ref{cl:notgood_budgets}.
\end{proof}

We now turn to our second type of deviation, but we need to further restrict the set of items that players bid on.
In particular, we let $\Gamma_{i} = \left\{j ~~\middle|~~ q_{ij} \leq \frac{1}{\gamma}\right\}$, and define $G_i = J_i \cap \Gamma_i$.
We now define the $\gamma$-boosting deviation (integral part) as $\bids_2^{\infloor{i}}=(\pbidi,\bidsmi)$, where
in $\pbidi$ buyer $i$ bids on a random $\intapprox{\gamma\cdot q_{ij}}$-fraction of each item $j\in G_i$ with
price $\rpricej$, where $\gamma > 1$ is a constant to be determined later. Note that each $\bids_2^{\infloor{i}}$ deviation for every $i\in\good$ is feasible since
$\good \subseteq \good_1$. Similarly, we define the fractional part of the $\gamma$-boosting deviation as $\bids_2^{\{i\}}$,
which is also a feasible deviation since $\good \subseteq \good_2$.  Also, since players bid on items in $G_i \subseteq \Gamma_i$, we have
$\intapprox{\gamma \cdot q_{ij}} \leq 1$ (we also have $\fracapprox{\gamma \cdot q_{ij}} \leq 1$, which holds for all items by definition).
\begin{lem}[$\gamma$-boosting deviation]
The value derived by buyers in $\good$ is comparable to the Liquid Welfare obtained at equilibrium:
\label{lem:second_deviation}
\bee
\left(1-\frac{2\alpha}{\gamma(\alpha-1)}\right)\sum_{i\in\good} \sum_j \vij{}\cdot q_{ij}
\le
\alpha\cdot \revenue(\strats) + 2\cdot \LW(\strats) - \frac{1}{\gamma}\sum_{i \in \notgood_3}\budgeti.
\eee
\end{lem}

\begin{proof}
For the integral part of the $\gamma$-boosting deviation, we can obtain bounds via the Nash equilibrium condition and Claim~\ref{cl:copies}:
\begin{align*}
\sum_{i\in\good}\sum_{j}\vij{}\cdot q_{ij} &\ge \sum_{i\in\good}\Ex[\bids\sim\strats]{\utili(\bids)}
\ge \sum_{i \in\good} \Ex[\bidsmi\sim\stratsmi]{\utili\left(\bids_2^{\infloor{i}}\right)}
\nonumber\\
&\ge \sum_{i \in\good}\sum_{j \in G_i}\left(1-\frac{1}{\alpha}\right) \intapprox{\gamma \cdot q_{ij}}\left(\vij{} - \rpricej\right).
\end{align*}
Similarly, for the fractional part of the $\gamma$-boosting deviation we get:
\bee
\sum_{i\in\good}\sum_{j}\vij{}\cdot q_{ij} \ge \sum_{i \in\good} \Ex[\bidsmi\sim\stratsmi]{\utili\left(\bids_2^{\{i\}}\right)}
\ge \sum_{i \in\good}\sum_{j \in G_i}\left(1-\frac{1}{\alpha}\right)\fracapprox{\gamma \cdot q_{ij}}\left(\vij{} - \rpricej\right).
\eee
Together these two deviations give us
\begin{align}
\label{eq:nash_second_deviation}
2\sum_{i\in\good}\sum_{j}\vij{}\cdot q_{ij} &\ge
\sum_{i \in\good}\sum_{j \in G_i}\left(1-\frac{1}{\alpha}\right) \gamma\cdot q_{ij}\left(\vij{} - \rpricej\right).
% \nonumber\\
% &\ge \left(1-\frac{1}{\alpha}\right)
% \left(\gamma\sum_{i \in\good} \sum_{j \in G_i}\left(\vij{}-\rpricej\right)\cdot q_{ij} -  \sum_{i\in\good}\sum_{j \in G_i}\vij{}\max\left(\gamma\cdot q_{ij} - 1,0\right)\right).
\end{align}
%\begin{align}
%\label{eq:nash_second_deviation}
%2\sum_{i\in\good}\sum_{j}\vij{}\cdot q_{ij} &\ge
%\sum_{i \in\good}\sum_{j \in J_i}\left(1-\frac{1}{\alpha}\right) \min\left(\gamma\cdot q_{ij},1\right)\left(\vij{} - \rpricej\right)
%\nonumber\\
%&\ge \left(\frac{1}{2}-\frac{1}{2\alpha}\right)
%\left(\gamma\sum_{i \in\good} \sum_{j \in J_i}\vij{}\cdot q_{ij} -  \sum_{i\in\good}\sum_{j \in J_i}\vij{}\max\left(\gamma\cdot q_{ij} - 1,0\right)\right),
%\end{align}
%where the last inequality follows from the fact that $\vij{} \geq 2\rpricej$ for each $j \in J_i$.
We further estimate the term $\sum_{i \in\good} \sum_{j \in G_i}\left(\vij{}-\rpricej\right)\cdot q_{ij}$ on the RHS of Equation~\eqref{eq:nash_second_deviation}:
\begin{align}
\sum_{i\in\good} \sum_{j \in G_i}\left(\vij{}-\rpricej\right)\cdot q_{ij}
&=\sum_{i\in\good} \sum_j \vij{}\cdot q_{ij} - \sum_{i \in \good} \sum_{j \not\in G_i} \vij{}\cdot q_{ij}-
\sum_{i\in\good} \sum_{j \in G_i}\rpricej\cdot q_{ij}\nonumber\\
&\ge\sum_{i\in\good} \sum_j \vij{}\cdot q_{ij} - \sum_{i \in \good} \sum_{j \not\in J_i} \rpricej\cdot q_{ij}-
\sum_{i \in \good} \sum_{j \not\in \Gamma_i} \vij{}\cdot q_{ij} - \sum_{i\in\good} \sum_{j \in G_i}\rpricej\cdot q_{ij}\nonumber\\
&\geq \sum_{i\in\good} \sum_j \vij{}\cdot q_{ij} - \sum_{i \in \good} \sum_{j \not\in \Gamma_i} \vij{}\cdot q_{ij} - \sum_{j}\rpricej \sum_{i \in\good }q_{ij}\nonumber \\
&\geq \sum_{i\in\good} \sum_j \vij{}\cdot q_{ij} - \sum_{i \in \good} \sum_{j \not\in \Gamma_i} \vij{}\cdot q_{ij} - \alpha \cdot \revenue(\strats),
\label{eq:first_term_boosting}
\end{align}
where the first inequality holds as $\sum_{j \not\in G_i}\vij{}q_{ij} \leq \sum_{j \not\in J_i}\vij{}q_{ij} + \sum_{j \not\in \Gamma_i}\vij{}q_{ij}$
and $\vij{} < \rpricej$ for every $j\notin J_i$.  Our next goal will be to bound the term
$\sum_{i\in\good}\sum_{j \not\in \Gamma_i}\vij{} \cdot q_{ij}$ on the RHS of Equation~\eqref{eq:first_term_boosting}. Before that we need to do
some preparations.  To ease the notations we denote by $\sharej$ the $\ell^{th}$ share of item $j$.
We observe that the expected Liquid Welfare at equilibrium is at least
\begin{align}
\label{eq:welfare_budgets}
\LW(\strats) &= \sum_{i} \prob[\bids\sim\strats]{\vali(\alloci)>\budgeti}\cdot\budgeti
+\sum_{i}\sum_{j}\sum_{\ell=1}^{\ncopy}\prob[\bids\sim\strats]{\{\vali(\alloci)\le\budgeti\}\land\{i \text{ wins }\sharej\}}\cdot\frac{\vij{}}{\ncopy}
\nonumber\\
&=\sum_{i} Q_i\cdot\budgeti+
\sum_{i,j}\vij{}\left(
\frac{1}{\ncopy}\sum_{\ell=1}^{\ncopy}\prob{i \text{ wins }\sharej}-\frac{1}{\ncopy}\sum_{\ell=1}^{\ncopy}\prob{\{\vali(\alloci)>\budgeti\}\land\{i \text{ wins }\sharej\}}
\right)
\nonumber\\
&=\sum_{i} Q_i\cdot\budgeti+ \sum_{i,j}\vij{}\cdot
\max\left\{0~,~q_{ij}-\frac{1}{\ncopy}\sum_{\ell=1}^{\ncopy}\prob{\{\vali(\alloci)>\budgeti\}\land\{i \text{ wins }\sharej\}}\right\}\nonumber\\
&\ge \sum_{i} Q_i\cdot\budgeti+ \sum_{i,j}\vij{}\cdot\max\left\{0,q_{ij}-\prob{\vali(\alloci)>\budgeti}\right\}\nonumber\\
&= \sum_{i} Q_i\cdot\budgeti +\sum_{i,j}\max\left\{0,q_{ij}-Q_i\right\}\cdot\vij{}
\geq \sum_{i \in \notgood_3}\frac{1}{2\gamma} \cdot \budgeti + \sum_{i \in \good}\sum_{j \not\in \Gamma_i}\vij{} \cdot \frac{q_{ij}}{2},
\end{align}
where the third equality holds true as the expression inside the $\max$ cannot be negative and $q_{ij}=\frac{1}{\ncopy}\sum_{\ell=1}^{\ncopy}\prob{i \text{ wins }\sharej}$ by definition of $q_{ij}$,
the first inequality holds since $\prob{\vali(\alloci)>\budgeti}\ge\prob{\{\vali(\alloci)>\budgeti\}\land\{i \text{ wins }\sharej\}}$, the last equality holds true by definition of $Q_i$,
and the last inequality holds since players $i \in \notgood_3$ have $Q_i > \frac{1}{2\gamma}$, while for players $i \in \good \subseteq \good_3$ and items $j \not\in \Gamma_i$
we have $Q_i \leq \frac{1}{2} \cdot \frac{1}{\gamma} \leq \frac{q_{ij}}{2}$.
Now we rearrange terms from Equation~\eqref{eq:welfare_budgets} to get
$\sum_{i \in \good}\sum_{j \not\in \Gamma_i}\vij{}q_{ij} \leq 2\cdot\LW(\strats) - \frac{1}{\gamma}\sum_{i \in \notgood_3}\budgeti$.
Combining Equation~\eqref{eq:nash_second_deviation} and Equation~\eqref{eq:first_term_boosting}, we can substitute this
upper bound to get:
\bee
2\sum_{i\in\good}\sum_{j}\vij{}\cdot q_{ij} \ge
\left(1-\frac{1}{\alpha}\right)\gamma
\left(\sum_{i\in\good} \sum_j \vij{}\cdot q_{ij} - \alpha \cdot \revenue(\strats) -
2 \cdot \LW(\strats) + \frac{1}{\gamma}\sum_{i \in \notgood_3}\budgeti\right).
\eee
Dividing both sides by $\left(1-\frac{1}{\alpha}\right)\gamma$ and rearranging terms gives the lemma:
\bee
\left(1-\frac{2\alpha}{\gamma(\alpha-1)}\right)\sum_{i\in\good} \sum_j \vij{}\cdot q_{ij}
\le
\alpha\cdot \revenue(\strats) + 2\cdot \LW(\strats) - \frac{1}{\gamma}\sum_{i \in \notgood_3}\budgeti.
\qedhere
\eee
\end{proof}

Now we have all necessary components to conclude the proof of Theorem~\ref{thm:mixed_first_price}
and show that the Liquid Price of Anarchy of any mixed Nash equilibrium is bounded.  
%The proof is given in Appendix~\ref{sec:appendix_misc}.

\begin{proof}[Proof of Theorem~\ref{thm:mixed_first_price}]
We combine the bounds from Lemma~\ref{lem:first_deviation} and Lemma~\ref{lem:second_deviation} and obtain
\begin{align*}
\alpha \cdot \revenue(\strats) &+ 2 \cdot \LW(\strats) - \frac{1}{\gamma} \sum_{i \in \notgood_3}\budgeti
\ge \\
&\left(1-\frac{2\alpha}{\gamma(\alpha-1)}\right)\left(\frac{1}{2}-\frac{1}{2\alpha}\right)
\left(\opt - \alpha\left(1+\gamma+\frac{n}{\ncopy}\right)\revenue(\strats) - \sum_{i \in \notgood_3}\budgeti \right).
\end{align*}
Since $\LW(\strats)\ge\revenue(\strats)$ we further derive that
\begin{align*}
\left(\alpha + 2 + \frac{1}{2}\left(1-\frac{1}{\alpha}-\frac{2}{\gamma}\right)\alpha\left(1+\gamma+\frac{n}{\ncopy}\right)\right) \LW(\strats) \ge \hspace{1in}\\
\frac{1}{2}\left(1-\frac{1}{\alpha}-\frac{2}{\gamma}\right)\opt + \left(\frac{1}{\gamma} - \frac{1}{2}\left(1-\frac{1}{\alpha}-\frac{2}{\gamma}\right)\right)\sum_{i \in \notgood_3}\budgeti.
\end{align*}

As long as the factor in front of $\sum_{i \in \notgood_3}\budgeti$ is nonnegative,
we have $\opt\le O\left(\frac{n}{\ncopy}\right)\cdot\LW(\strats)$ for any $1 \leq \ncopy \leq n$ for a particular choice of
parameters (e.g., $\alpha = 2.26,\gamma=7.16$). In particular, when $\ncopy\ge n$, we have that the $\lpoa$ is at most $51.5$.
\end{proof}

\noindent \textit{Remark~\ref{rem:clear_house}.}
The bound on the Liquid Price of Anarchy derived in Theorem~\ref{thm:mixed_first_price} holds for
simultaneous first price auctions with the house clearing item bidding mechanism.
\begin{proof}
We observe that the bidding strategy defined in Claim~\ref{cl:copies} extends naturally to the new house clearing mechanism.
Moreover, for a new equilibrium $\strats$ with appropriately redefined item prices $\rprices(\strats)$, the argument from
the proof of Claim~\ref{cl:copies} gives us exactly the same guarantee on the expected number of shares won by a bidder.
Indeed, in the new mechanism when a bidder faces competitors' bids or equivalently the set of share prices $\{\pricejl[j]{\ell}\}_{\ell=1}^{\ncopy}$,
they avoid the uncertainty of bidding on all shares $\ell$ with high prices $\frac{\rpricej}{\ncopy}<\pricejl[j]{\ell}$ and thus they
receive at least as many copies as they would get in the independent share bidding auction.

We further note that all the remaining parts of the proof of Theorem~\ref{thm:mixed_first_price} do not depend on the market clearing format of first price item bidding auctions.
\end{proof}

\section{Pure Nash Equilibria}
\label{sec:lpoa}
%\section{Liquid Price of Anarchy for Simultaneous First Price Auctions}
%\label{sec:pure_first}
%As a warmup we begin by studying efficiency of pure Nash equilibria of the first price auction.
We also study {\em pure} Nash equilibria of simultaneous first price auctions with deterministic tie-breaking rules.
The proofs of the next two theorems are given in Appendices~\ref{sec:appendix_pure} and~\ref{sec:appendix_lower},
respectively.
%~\ref{thm:1paupper} and \ref{thm:lb_rand_ties}
%\subsection{Warmup: Pure Nash Equilibria}

%\begin{thm}
%\label{thm:warmup}
%If $\bids$ is a pure Nash equilibrium, then $LW(\bids) \geq \frac{OPT}{2} - \eps$ for arbitrary small $\eps > 0$.
%\end{thm}
%\begin{proof}
%\todo{simple proof for additive buyers.}
%\end{proof}
%
%The same result still holds for buyers with more complex combinatorial preferences. We defer the
%proof to Appendix~\ref{sec:appendix_pure}.

\begin{thm}\label{thm:1paupper}
Consider a simultaneous first price auction where budgeted bidders have fractionally-subadditive
valuations\footnote{Valuation $\val$ is fractionally-subadditive or equivalently XOS if there is a
set of additive valuations $A = \{a_1,\ldots,a_\ell\}$ such that $\vali(S) = \max_{a\in A} a(S)$
for every $S \subseteq [m]$. XOS is a super class of submodular and additive valuations.}.
If $\bids$ is a pure Nash equilibrium, then $\LW(\bids) \ge \frac{\opt}{2}$.
\end{thm}
\noindent A complementary tightness result for Theorem~\ref{thm:1paupper} is given in Appendix~\ref{sec:pureExamples}.

Unfortunately, this result is not quite satisfying compared to mixed Nash equilibria. 
The first reason is that pure Nash equilibria might not even exist.
%, which
%we show in Appendix~\ref{sec:appendix_pure} (even if items can be divided
%into arbitrarily many parts).
Consider the following auction with $n=2$ players and $m=10$ identical items,
where player $1$ values each item at $1$ and has a budget of $1$, while the second player values each item at $1.1$,
and has a budget of $1.1$. It is easy to see that there can be no pure Nash equilibrium, for the following
reason.  If player $2$ knows what player $1$ is bidding, player $2$ has the budget to simply outbid player $1$ and
win all items.  On the other hand, if player $2$ is winning all items, one of their bids must be at most $0.11$ (due
to their budget constraint), and therefore player $1$ can outbid player $2$ on this item.  Hence, there is no
pure Nash equilibrium.  
%Note that this example can easily be extended to the case where all items can be divided into
%arbitrarily many parts.  
Moreover, it holds true for both simultaneous first price and second price auctions.
Second, the LPoA guarantees strongly depend on the tie-breaking rules used in the auction. In particular,
if we allow randomized tie-breaking rules, the LPoA is no longer a constant.

\begin{thm}\label{thm:lb_rand_ties}
With a randomized tie-breaking rule, there are simultaneous first and second price auction games
which have an $\Omega(n)$ Liquid Price of Anarchy, even when agents play pure strategies.
\end{thm}

%\section{Liquid Price of Anarchy for Simultaneous Second Price Auctions}
%\label{sec:pure_second}

%
%
%\begin{thm}\label{thm:2paupper}
%Consider a simultaneous second price auction in which bidders have fractionally-subadditive valuations.  If $\bids$ is a pure Nash
%equilibrium where players' bids satisfy no over-bidding and no over-budgeting, then we have $LW(\bids) \geq \frac{OPT}{2}$.
%\end{thm}

%\section{Non-existence of Pure Nash Equilibria}
%\paragraph{Non-existence of Pure Nash Equilibria}
%We give an example of a simultaneous first price auction game for which there is no pure Nash equilibrium (the example also
%holds for simultaneous second price auctions).  Consider the following game with $n=2$ players and $m=2$ items,
%where player $1$ values both items at $9$ and has a budget of $10$, while the second player values item $1$ at $10$,
%item $2$ at $16$, and has a budget of $15$.  It is easy to see that there can be no pure Nash equilibrium, for the following
%reason.  If player $2$ knows what player $1$ is bidding, player $2$ has the budget to simply outbid player $1$ and
%win both items.  On the other hand, if player $2$ is winning both items, one of the bids must be at most $7.5$ (due
%to their budget constraint), and hence player $1$ can outbid player $2$ on this item.  Hence, there is no
%pure Nash equilibrium. 

%\subsection{Multi-Share Auctions}
%\label{sec:share_motivation}
%\subsection{Share Model Motivation}

\paragraph{Lower Bound: Divisible Items.}
In order to reconcile this big gap between the Liquid Welfare of the optimal allocation and 
pure equilibria with randomized tie-breaking rules
(and mixed Nash equilibria in general), we considered in Section~\ref{sec:main}
a share model in which each item $j$ is divided into $\ncopy$ identical shares and obtained an upper bound of 
$O\left(1+\frac{n}{\ncopy}\right)$ on the LPoA. We now show a complementary lower bound that the Liquid Price of Anarchy
is super constant for such equilibria when the number of shares is $\ncopy = o(n)$. The proof of Theorem~\ref{thm:lb_mixed_shares}
is deferred to Appendix~\ref{sec:appendix_lower}.

\begin{thm}
\label{thm:lb_mixed_shares}
There are some simultaneous first price and second price auction games
for which the Liquid Price of Anarchy is $\Omega\left(\frac{n}{\ncopy}\right)$ when the number of shares is $\ncopy$.
\end{thm}

%In order to reconcile this big gap between the Liquid Welfare of the optimal allocation and mixed Nash equilibria (and
%pure equilibria with randomized tie-breaking rules), we consider
%a share model in which each item $j$ is divided into $\ncopy$ identical shares, where player $i$ values each
%share at $\frac{\vij{}}{\ncopy}$ (here, as our examples are in the additive setting, $\vij{}$ denotes player $i$'s value for item $j$).
%%We imagine a scenario in which each buyer specifies how many shares they want to buy and the price they are willing
%%to bid per share.  Item $j$'s shares are distributed to agents according to the bids submitted for $j$ in decreasing order.
%We now show that the Liquid Price of Anarchy is super-constant for such equilibria when the number of shares is $\ncopy = o(n)$.
%The proof of Theorem~\ref{thm:lb_mixed_shares} is deferred to Appendix~\ref{sec:appendix_lower}.
%
%\begin{thm}
%\label{thm:lb_mixed_shares}
%There are some simultaneous first price and second price auction games
%for which the Liquid Price of Anarchy is $\Omega\left(\frac{n}{\ncopy}\right)$ when the number of shares is $\ncopy$.
%\end{thm}

\section{Extensions}
\label{sec:extensions}
\begin{rem}\label{rem:clear_house}
The bound on the Liquid Price of Anarchy derived in Theorem~\ref{thm:mixed_first_price} holds for
simultaneous first price auctions with the house clearing item bidding mechanism.
\end{rem}
%
%More discussion on the house clearing mechanism is given in Appendix~\ref{sec:appendix_misc}.
%\begin{proof}
%We observe that the bidding strategy defined in Claim~\ref{cl:copies} extends naturally to the new house clearing mechanism. Moreover, for a new equilibrium $\strats$ with appropriately redefined item prices $\rprices(\strats)$, the argument from the proof of Claim~\ref{cl:copies} gives us exactly the same guarantee on the expected number of shares won by a bidder. Indeed, in the new mechanism when a bidder faces competitors bids or equivalently the set of share prices $\{\pricejl[j]{\ell}\}_{\ell=1}^{\ncopy}$, she avoids uncertainty of bidding on all shares $\ell$ with high prices $\frac{\rpricej}{\ncopy}<\pricejl[j]{\ell}$ and thus she receives at least as many copies as she would get in the independent share-bidding auction.
%
%We further note that all the remaining parts of the proof of Theorem~\ref{thm:mixed_first_price} do not depend on the market clearing format of the first-price item-bidding auction.
%\end{proof}
The result of Theorem~\ref{thm:mixed_first_price} also extends to Bayesian first price auctions and very similar results hold true for
the second price auction format, under standard necessary restrictions on the bidding strategy of the buyers (we describe and discuss
these restrictions in Appendix~\ref{sec:appendix_sp}).  The formal proof of Theorem~\ref{thm:mixed_first_price_bayesian} is given in
Appendix~\ref{sec:appendix_bayesian} and closely follows the proof of Theorem~\ref{thm:mixed_first_price}.  
%We also defer the proof of Theorem~\ref{thm:mixed_second_price_bayesian} to Appendix~\ref{sec:appendix_misc}.

\begin{thm}
\label{thm:mixed_first_price_bayesian}
In simultaneous first price auctions with $n$ additive bidders and budgets, where every item
has $\ncopy$ equal shares (copies), the Liquid Price of Anarchy of Bayesian Nash equilibria is
$O\left(1+\frac{n}{\ncopy}\right)$ (at most $51.5$, when $\ncopy\ge n$).
\end{thm}

%For the second price auction we have very similar results, under standard necessary restrictions on the bidding strategy of the buyers.

\begin{thm}
\label{thm:mixed_second_price_bayesian}
In simultaneous second price auctions with $n$ additive bidders and budgets, where every item
has $\ncopy$ equal shares (copies), the Liquid Price of Anarchy of Bayesian Nash equilibria is
$O\left(1+\frac{n}{\ncopy}\right)$ (at most $51.5$, when $\ncopy\ge n$), under the
no over-bidding and no over-budgeting assumptions.
\end{thm}
\begin{proof}
The proof that Bayesian Nash equilibria achieve a good Liquid Price of Anarchy for simultaneous second price auctions
follows similarly to the proof of Theorem~\ref{thm:mixed_first_price_bayesian} in Appendix~\ref{sec:appendix_bayesian}.
In the following we only discuss the differences that are related to the second price format. 
%There are a few differences that should be noted.
In first price simultaneous auctions we have $\revenue(\strats) = \frac{1}{\alpha}\sum_{j=1}^{m}\rpricej$, while in the second price auction
each bidder pays less than $\frac{1}{\alpha}\rpricej$ per item $j$. However, the proof for simultaneous second price auctions 
goes through if we substitute $\revenue(\strats)$ with the sum of prices $\frac{1}{\alpha}\sum_{j=1}^{m}\rpricej$, 
for the reason that the property $\frac{1}{\alpha}\sum_{j=1}^{m}\rpricej \le \LW(\strats)$ still holds.
To see why, fix a pure bidding profile $\bids = (\bidi[1],\ldots,\bidi[n])$ of the players.  
Recall that $\sharej$ denotes the $\ell^{th}$ share of item $j$, and that
$\alloci(\bids)$ denotes the set of shares of items that player $i$ wins when players bid according to $\bids$.
In particular, the no over-bidding and no over-budgeting
assumptions (see Appendix~\ref{sec:appendix_sp} for the definitions of these assumptions)
imply $\sum_{\sharej \in \alloci(\bids)}\bid[i]{j}^\ell \leq \vali(\alloci(\bids))$ and
$\sum_{\sharej \in \alloci(\bids)}\bid[i]{j}^\ell \leq B_i$.  Hence, similarly to the simultaneous first price auction setting, we get:
\begin{align*}
\revenue(\strats) = \frac{1}{\alpha}\sum_{j=1}^{m}\rpricej &= \sum_{j=1}^m\sum_{\ell=1}^{\ncopy} \Ex{\pricejl[j]{\ell}}
= \sum_{j=1}^m\sum_{\ell=1}^{\ncopy} \Ex{\max_i \bid[i]{j}^\ell}\\
&= \sum_{i=1}^n\Ex{\sum_{\sharej \in \alloci(\bids)}\bid[i]{j}^\ell} \leq \sum_{i=1}^n\Ex{\min\{\vali(\alloci(\bids)),B_i\}} = \LW(\strats),
\end{align*}
where the expectation is taken over players' valuation profiles and randomized bidding strategies.

The last observation we mention is that the utility a player receives when performing a deviation in simultaneous
second price auctions is at least as high as the utility they would receive if they were forced to pay their bid (since
players pay the second highest bid).  Thus, all inequalities involving the utility each player derives when performing
a deviation still hold.  It is not hard to verify that the rest of the proof goes through as well.
\end{proof}

%\todo{add remark about more realistic item clearing mechanism. Comment that all our LPoA bounds hold unchanged. The reason is that we use only deviation as described in the Claim~\ref{cl:copies} for given  expected price per share, a player wins more items compared to the independent bidding. The intuition is that given the set of competitive prices on the shares of a particular item, the player avoids the uncertainty of random bidding and only competes with the lower end prices.}

%%%%%%%%%%%%%%%%%%%%%%%%%%%%%%%%%%%%%%%%%%%%%%%%%%%%%%%%%%%%%%%%%%%%%%%%%%%%%%%%%%%%%%%%%%%%%%%%%%%%%%%%%%%%%%%%%%%%%%%%%%%

%\input{mixed_nick}

\bibliographystyle{plain}
\bibliography{LPOA}

\newpage
\appendix

\section{Second Price Auctions.}
\label{sec:appendix_sp}

We also study {\em simultaneous second price auctions}, where
each item $j$ again is allocated to the highest bidder, but the winner pays the second highest bid on item $j$.
The total payment of bidder $\ell$ in this case is $p_\ell(\bids) = \sum_{j \in \alloci(\bids)}\max_{i\neq \ell}\bid[i]{j}$.
As before one can study the Liquid Price of Anarchy for the second price format.
However, in general, the Price of Anarchy of second price auctions can be arbitrarily large even when bidders do not have budget constraints
and there is only one item for sale\footnote{A canonical example is two bidders who value the item at $0$ and a
large number~$L$, respectively, but the first bidder bids $L + 1$ and the second bidder bids~$0$.}. To prevent such
pathological equilibria, it is standard~\cite{CKS08, BR11, FFGL13} to assume that each bidder is guaranteed to derive
non-negative utility, no matter how the other bidders behave, i.e., $\sum_{j \in S} \bid[i]{j} \le \vali(S)$ $\forall S \subseteq [m]$.
In other words, no bidder would want to ``overbid'' on any set of items.
In the budgeted setting, in addition to this {\em no over-bidding} assumption, we also require that
$\sum_{j \in [m]}\bid[i]{j} \leq \budgeti$ (i.e., no \emph{over-budgeting}). The latter is a necessary assumption even in the
single-item case to exclude pathological equilibria. The no over-budgeting assumption can also be motivated by risk-averse attitudes
of buyers, who try to eliminate any chance of exceeding their budget and deriving infinite disutility.

\section{Pure Nash Equilibria: Beyond Additive Valuations}
\label{sec:appendix_pure}
In this section we consider buyers with complex combinatorial valuations and study the efficiency of pure Nash equilibria of
simultaneous item bidding auctions for bidders with budgets. We first show that the Liquid Price of Anarchy for second price auctions, in which bidders have
fractionally-subadditive valuations, is $2$. The proof is inspired by~\cite{CKS08}. We recall the definition of fractionally-subadditive valuations. 

\begin{defn}%[Fractionally-Subadditive Valuation]
Valuation $\vali$ is fractionally-subadditive if there is a set of additive valuations
$A = \{a_1,\ldots,a_\ell\}$ (for some $\ell \geq 0$) such that $\vali(S) = \max_{a\in A} a(S)$
for every $S \subseteq [m]$.
\end{defn}
\noindent Valuation $a\in A$ is called a \emph{maximizing additive valuation} for
a particular set $S$ if $\vali(S) = a(S)$.

\begin{thmappsec}
\label{thm:2paupper}
Consider a simultaneous second price auction in which bidders have fractionally-subadditive valuations.  If $\bids$ is a pure Nash
equilibrium where players' bids satisfy no over-bidding and no over-budgeting, then we have $\LW(\bids) \geq \frac{\opt}{2}$.
\end{thmappsec}
\begin{proof}
We begin with the following useful lemma.
\begin{lem}\label{lem:devlb2pa}
Fix an arbitrary $S \subseteq [m]$ and a player $i$ such that $v_i(S) > 0$, and let $a_r$ be a maximizing additive valuation
for $S$.  Consider the alternative bidding strategy $\abidi$ for $i$, where $\abid[i]{j} = a_r(\{j\}) \frac{\min\{v_i(S),\budgeti \}}{v_i(S)}$ for $j \in S$
and $\abid[i]{j} = 0$ for $j \not\in S$.  Then for any pure profile $\bidsmi$ we have:

$$ u_i(b_i^*,\bidsmi) \geq \min\{v_i(S),\budgeti\} - \sum_{j \in S} \max_{k \neq i} b_{kj}.$$
\end{lem}
\begin{proof}
Let $T$ be the set of items that player $i$ wins in the allocation $\alloci(b_i^*,\bidsmi)$.  Note that
$\max_{k \neq i}b_{kj} = 0$ for any $j \in T \setminus S$ and
$\abid[i]{j} - \max_{k \neq i}b_{kj} = a_r(\{j\})\frac{\min\{v_i(S),\budgeti\}}{v_i(S)} - \max_{k \neq i}b_{kj} \leq 0$
for any $j \in S \setminus T$.  Then, we have:
\begin{eqnarray*}
u_i(b_i^*,\bidsmi) = v_i(T) - \sum_{j \in T}\max_{k \neq i}b_{kj} &\geq& \sum_{j \in T \cap S}a_r(\{j\}) - \sum_{j \in T \cap S}\max_{k \neq i}b_{kj} \\
&\geq& \sum_{j \in T \cap S}a_r(\{j\})\frac{\min\{v_i(S),\budgeti\}}{v_i(S)} - \sum_{j \in T \cap S}\max_{k \neq i}b_{kj}\\
&\geq& \sum_{j \in S}a_r(\{j\})\frac{\min\{v_i(S),\budgeti\}}{v_i(S)} - \sum_{j \in S}\max_{k \neq i}b_{kj} \\
&=& \min\{v_i(S),\budgeti\} - \sum_{j \in S}\max_{k \neq i}b_{kj}.
\end{eqnarray*}
\end{proof}

%\label{thm:2paupper}

Let $v_1,\ldots,v_n$ denote the valuations of the players, and fix a particular player $i$.  Let $\opt$ denote the value of the optimal
solution, and let $S_i^*$ denote the set of items that player $i$ receives in an optimal allocation and $S_i = \alloci(\bids)$
be $i$'s allocation in the pure Nash equilibrium.  Suppose that $v_i(S_i^*) > 0$ (we will
handle the case that $v_i(S_i^*) = 0$ separately).  Let $a_r$ be the maximizing additive valuation for the set $S_i^*$, and consider
the following deviating strategy for player $i$: $\abid[i]{j} = a_r(\{j\}) \frac{\min\{v_i(S_i^*),\budgeti \}}{v_i(S_i^*)}$ for $j \in S_i^*$
and $\abid[i]{j} = 0$ for $j \not\in S_i^*$.  Notice that the alternative strategy $b_i^*$ satisfies no over-bidding and no over-budgeting.

Then, by Lemma~\ref{lem:devlb2pa}, we have
$u_i(b_i^*,\bidsmi) \geq \min\{v_i(S_i^*),\budgeti\} - \sum_{j \in S_i^*} \max_{k \neq i} b_{kj}$.
Since $\bids$ is a pure Nash equilibrium, the following must hold:
$$
v_i(S_i) \geq u_i(b_i,\bidsmi) \geq u_i(b_i^*,\bidsmi) \geq \min\{v_i(S_i^*),\budgeti\} - \sum_{j \in S_i^*} \max_{k \neq i} b_{kj}.
$$
Hence, if $v_i(S_i^*) > 0$ and $v_i(S_i) \leq \budgeti$, then we have shown that
$\min\{v_i(S_i),\budgeti\} = v_i(S_i) \geq \min\{v_i(S_i^*),\budgeti\} - \sum_{j \in S_i^*} \max_{k \neq i} b_{kj}$.
Now, consider a player $i$ such that $v_i(S_i^*) = 0$ and $v_i(S_i) \leq \budgeti$.  Then again we have
$\min\{v_i(S_i),\budgeti\} \geq \min\{v_i(S_i^*),\budgeti\} \geq \min\{v_i(S_i^*),\budgeti\} - \sum_{j \in S_i^*} \max_{k \neq i} b_{kj}$.
Finally, consider a player $i$ such that $v_i(S_i) > \budgeti$.  In such a case, we have
$\min\{v_i(S_i),\budgeti\} = \budgeti \geq \min\{v_i(S_i^*),\budgeti\} \geq \min\{v_i(S_i^*),\budgeti\} - \sum_{j \in S_i^*} \max_{k \neq i} b_{kj}$.
Hence, in each case, we have the same lower bound on $\min\{v_i(S_i),\budgeti\}$.  Putting these together and
summing over all bidders, we get:
\begin{eqnarray*}
\LW(\bids) = \sum_{i=1}^n \min\{v_i(S_i),\budgeti\} &\geq& \sum_{i=1}^n\min\{v_i(S_i^*),\budgeti\} - \sum_{i=1}^n\sum_{j \in S_i^*} \max_{k \neq i} b_{kj} \\
&\geq& \opt - \sum_{j=1}^m \max_{k} b_{kj} \\
&=& \opt - \sum_{i=1}^n\sum_{j \in S_i}\bid[i]{j} \\
&\geq& \opt - \sum_{i=1}^n\min\{v_i(S_i),\budgeti\} = \opt - \LW(\bids),
\end{eqnarray*}
where the last inequality follows since $\bids$ satisfies no over-bidding (i.e., $\sum_{j \in S_i} \bid[i]{j} \leq v_i(S_i)$) and
no over-budgeting (i.e., $\sum_{j \in S_i}\bid[i]{j} \leq \budgeti$).
\end{proof}

Similarly, the Liquid Price of Anarchy for pure Nash equilibria of simultaneous first-price auctions, in which bidders have
fractionally-subadditive valuations, is arbitrarily close to $2$.

\begin{thmappsec}[Theorem~\ref{thm:1paupper}]
\label{thm:1paupper_app}
Consider a simultaneous first price auction where bidders have fractionally-subadditive valuations.  If $\bids$ is a pure Nash
equilibrium, then for any $\epsilon > 0$ we have $\LW(\bids) \geq \frac{\opt}{2} - \epsilon$.
\end{thmappsec}
\begin{proof}
We begin with the following statement, which is similar to Lemma~\ref{lem:devlb2pa}.
\begin{lem}\label{lem:devlb1pa}
Fix an arbitrary $S \subseteq [m]$ and a player $i$ such that $v_i(S) > 0$, and let $a_r$ be a maximizing additive valuation
for $S$.  For any $\delta > 0$, consider the alternative bidding strategy $b_i^*$ for $i$, where
$\abid[i]{j} = \min\{a_r(\{j\}) \frac{\min\{v_i(S),\budgeti \}}{v_i(S)},\max_{k \neq i}b_{kj} + \delta\}$ for $j \in S$
and $\abid[i]{j} = 0$ for $j \not\in S$.  Then for any pure profile $\bidsmi$ we have:

$$ u_i(b_i^*,\bidsmi) \geq \min\{v_i(S),\budgeti\} - \sum_{j \in S} \max_{k \neq i} b_{kj} - \delta |S|.$$
\end{lem}
\begin{proof}
Let $T$ be the set of items that player $i$ wins in the allocation $\alloci(b_i^*,\bidsmi)$.  Note that
for any $j \in T \setminus S$, we have $\abid[i]{j} = 0$ and for any $j \in T \cap S$, we have $\abid[i]{j} \leq \max_{k \neq i}b_{kj} + \delta$.
Moreover, $\abid[i]{j} - \max_{k \neq i}b_{kj} \leq 0$ and $\abid[i]{j} = a_r(\{j\})\frac{\min\{v_i(S),\budgeti\}}{v_i(S)}$
for any $j \in S \setminus T$.  Then, we have:
\begin{eqnarray*}
u_i(b_i^*,\bidsmi) = v_i(T) - \sum_{j \in T}\abid[i]{j} &\geq& \sum_{j \in T \cap S}a_r(\{j\}) - \sum_{j \in T \cap S}\max_{k \neq i}b_{kj} - \delta |T \cap S| \\
&\geq& \sum_{j \in T \cap S}a_r(\{j\})\frac{\min\{v_i(S),\budgeti\}}{v_i(S)} - \sum_{j \in T \cap S}\max_{k \neq i}b_{kj} - \delta |S|\\
&\geq& \sum_{j \in S}a_r(\{j\})\frac{\min\{v_i(S),\budgeti\}}{v_i(S)} - \sum_{j \in S}\max_{k \neq i}b_{kj} - \delta |S| \\
&=& \min\{v_i(S),\budgeti\} - \sum_{j \in S}\max_{k \neq i}b_{kj} - \delta |S|. 
\end{eqnarray*}
\end{proof}

Let $v_1,\ldots,v_n$ denote the valuations of the players, and fix a particular player $i$.  Let $\opt$ denote the value of the optimal
solution, and let $S_i^*$ denote the set of items that player $i$ receives in an optimal allocation and $S_i = \alloci(\bids)$
be $i$'s allocation in the pure Nash equilibrium.  Suppose that $v_i(S_i^*) > 0$ (we will
handle the case that $v_i(S_i^*) = 0$ separately).  Let $a_r$ be the maximizing additive valuation for the set $S_i^*$, and consider
the following deviating strategy for player $i$ for some $\delta > 0$ to be set later:
$\abid[i]{j} = \min\{a_r(\{j\}) \frac{\min\{v_i(S),\budgeti \}}{v_i(S)},\max_{k \neq i}b_{kj} + \delta\}$ for $j \in S_i^*$
and $\abid[i]{j} = 0$ for $j \not\in S_i^*$.  Notice that the alternative strategy $\abidi$ satisfies no over-budgeting.

Then, by Lemma~\ref{lem:devlb1pa}, we have
$u_i(\abidi,\bidsmi) \geq \min\{v_i(S_i^*),\budgeti\} - \sum_{j \in S_i^*} \max_{k \neq i} b_{kj} - \delta |S_i^*|$.
Since $\bids$ is a pure Nash equilibrium, the following must hold:
$$
v_i(S_i) \geq u_i(b_i,\bidsmi) \geq u_i(\abidi,\bidsmi) \geq \min\{v_i(S_i^*),\budgeti\} - \sum_{j \in S_i^*} \max_{k \neq i} b_{kj} - \delta |S_i^*|.
$$
Hence, if $v_i(S_i^*) > 0$ and $v_i(S_i) \leq \budgeti$, then we have shown that
$\min\{v_i(S_i),\budgeti\} = v_i(S_i) \geq \min\{v_i(S_i^*),\budgeti\} - \sum_{j \in S_i^*} \max_{k \neq i} b_{kj} - \delta |S_i^*|$.
Now, consider a player $i$ such that $v_i(S_i^*) = 0$ and $v_i(S_i) \leq \budgeti$.  Then again we have
$\min\{v_i(S_i),\budgeti\} \geq \min\{v_i(S_i^*),\budgeti\} \geq \min\{v_i(S_i^*),\budgeti\} - \sum_{j \in S_i^*} \max_{k \neq i} b_{kj} - \delta |S_i^*|$.
Finally, consider a player $i$ such that $v_i(S_i) > \budgeti$.  In such a case, we have
$\min\{v_i(S_i),\budgeti\} = \budgeti \geq \min\{v_i(S_i^*),\budgeti\} \geq \min\{v_i(S_i^*),\budgeti\} - \sum_{j \in S_i^*} \max_{k \neq i} b_{kj} - \delta |S_i^*|$.
Hence, in each case, we have the same lower bound on $\min\{v_i(S_i),\budgeti\}$.  Putting these together and
summing over all bidders, we get:
\begin{eqnarray*}
\LW(\bids) = \sum_{i=1}^n \min\{v_i(S_i),\budgeti\} &\geq& \sum_{i=1}^n\min\{v_i(S_i^*),\budgeti\} - \sum_{i=1}^n\sum_{j \in S_i^*} \max_{k \neq i} b_{kj} - \delta \sum_{i=1}^n|S_i^*| \\
&\geq& \opt - \sum_{j=1}^m \max_{k} b_{kj} - \delta m \\
&=& \opt - \sum_{i=1}^n\sum_{j \in S_i}\bid[i]{j} - \delta m \\
&\geq& \opt - \sum_{i=1}^n\min\{v_i(S_i),\budgeti\} - \delta m = \opt - \LW(\bids) - \delta m,
\end{eqnarray*}
where the last inequality follows since in any Nash equilibrium, the bids $\bids$ must satisfy $\sum_{j \in S_i} \bid[i]{j} \leq v_i(S_i)$
and $\sum_{j \in S_i}\bid[i]{j} \leq \budgeti$.  Noting that we can set $\delta = \frac{2\epsilon}{m}$ gives the theorem.
\end{proof}

%\subsection*{Non-existence of Pure Nash Equilibria}
%Consider the following auction with $n=2$ players and $m=10$ items,
%where player $1$ values each item at $1$ and has a budget of $1$, while the second player values each item at $1.1$,
%and has a budget of $1.1$. It is easy to see that there can be no pure Nash equilibrium, for the following
%reason.  If player $2$ knows what player $1$ is bidding, player $2$ has the budget to simply outbid player $1$ and
%win all items.  On the other hand, if player $2$ is winning all items, one of their bids must be at most $0.11$ (due
%to their budget constraint), and therefore player $1$ can outbid player $2$ on this item.  Hence, there is no
%pure Nash equilibrium.  Note that this example can easily be extended to the case where all items can be divided into
%arbitrarily many parts.  Moreover, it holds true for both simultaneous first price and second price auctions.

\section{Lower Bounds}
\label{sec:appendix_lower}
We give an example with additive valuations which shows that a randomized tie-breaking rule can lead to a Liquid Price of Anarchy which is $\Omega(n)$
in the following theorem.

\begin{thmappsec}[Theorem~\ref{thm:lb_rand_ties}]
With a randomized tie-breaking rule, there are simultaneous first and second price auction games
which have an $\Omega(n)$ Liquid Price of Anarchy, even when agents play pure strategies.
\end{thmappsec}
\begin{proof}
Consider a simultaneous first price auction in which there are $n$ items and $n$ bidders (the proof goes
through for simultaneous second price auctions as well).  Each bidder has the same
valuation profile for the items, namely they all value the first item at $n^4$ and the remaining $n-1$ items at $1$.  Moreover,
each bidder has a budget of $1$.  Observe that the optimal solution is $n$, which is attained by giving each agent
one item, for a total Liquid Welfare of $n$ (since each player is budget-capped at $1$).

On the other hand, consider the following randomized tie-breaking rule, where all ties among the highest bids that
occur for item $1$ are broken uniformly at random (the tie-breaking rule for the remaining items can be arbitrary).
Since item $1$ is valued so highly, the pure strategy profile in which each bidder bids $1$ (i.e., their maximum possible bid)
for item $1$ is a pure Nash equilibrium.  In particular, each player wins item 1 with probability $\frac{1}{n}$, and since
all $n$ agents are tied for the maximum bid, their utility is $\frac{1}{n} \cdot (n^4 - 1)$.  If they switch,
their utility will be at most $n-1$ (the same holds for second price auctions).  This leads to a Liquid Welfare of $1$, which implies the Liquid
Price of Anarchy is $\Omega(n)$.
\end{proof}

In fact, the same problem also exists when the tie-breaking rule is deterministic and players can randomize their strategies.

\begin{thmappsec}
\label{thm:lb_mixed}
With a deterministic tie-breaking rule, there are some simultaneous first price and second price auction games
for which the Liquid Price of Anarchy is $\Omega(n)$ when agents mix their strategies.
\end{thmappsec}
\begin{proof}
Consider a simultaneous first price auction in which there are $n$ items and $n$ bidders (the proof goes
through for simultaneous second price auctions as well), each of which has a budget of $1$.
Items $1$ and $2$ have a tie-breaking rule which favors players lexicographically (i.e., player $1$ is the most preferred while
player $n$ is the least preferred), while the rest of the items have an arbitrary tie-breaking rule.  All players value
the first two items at $2^{2n}$.  The first $n-2$ players value the remaining items $3,\ldots,n$ at $1$, while players $n-1$
and $n$ value items $3,\ldots,n$ at $0$.  Observe that the optimal solution is $n$, which is attained by allocating item $1$
to player $n-1$, item $2$ to player $n$, and giving each of the remaining $n-2$ items to each of the remaining $n-2$ players.

Consider the following bad mixed Nash equilibrium.  Players $n$ and $n-2$ have pure strategies where both of them bid
their full budget on item $1$ (recall all players have a budget of $1$).  Players $n-1$ and $n-3$ have pure strategies where
both of them bid their full budget on item $2$.  The first $n-4$ players randomize their strategies by bidding
their full budget on item $1$ with probability $\frac{1}{2}$ and bidding their full budget on item $2$ with probability
$\frac{1}{2}$.  This leads to a Liquid Welfare of $2$, since only two players win an item.  Hence, the Liquid Price
of Anarchy is $\Omega(n)$.

Let us see why this mixed strategy profile forms an equilibrium.  Players $n-1$ and $n$ never win an item since the
tie-breaking rule prefers player $n-3$ to player $n-1$ and prefers player $n-2$ to player $n$, but they cannot deviate
and improve their utility.  In particular, both players value items $n-3,\ldots,n$ at $0$, so they cannot bid on such
items.  Moreover, even if player $n$ deviates and bids their full budget on item $2$, they still lose since the
tie-breaking rule prefers player $n-3$ (similarly, player $n-1$ cannot improve their utility by deviating to item $1$
since the tie-breaking rule prefers player $n-2$).  Player $n-2$ currently wins exactly when players $1,\ldots,n-4$
all randomly choose item $2$, which happens with probability $\frac{1}{2^{n-4}}$, and hence has an expected utility
of at least $\frac{1}{2^{n-4}}(2^{2n} - 1) \geq 2^n$ (similarly, player $n-3$ wins when players $1,\ldots,n-4$ all
choose item $1$, and hence player $n-3$ has the same expected utility as player $n-2$).  Consider any other possible
pure strategy for player $n-2$.  They cannot win item $1$ since they bid strictly less than $1$, and they cannot win
item $2$ since player $n-3$ is preferred by the tie-breaking rule.  Their expected utility in any such deviation is
at most $n-2 < 2^n$, which is attained by winning each of the items valued at $1$.  A similar argument holds for player $n-3$, the
only difference being that if they deviate and bid their whole budget on item $1$, they can win the item (but this results
in the same expected utility).  Finally, observe that each player $i > n-3$ wins an item of value $2^{2n}$
when all players $1,\ldots,i-1$ bid on the other item of value $2^{2n}$, and hence player $i$'s expected utility is
at least $2^n$.  Any deviation will result in the same expected utility, or an expected utility of at most $n-2$.
Hence, we have a mixed Nash equilibrium.
\end{proof}

In the following theorems, we consider the affect of shares on the Liquid Price of Anarchy.  We write the proofs assuming
the market clearing mechanism for shares discussed in Section~\ref{sec:main}, but the proofs can easily be extended
to handle the share model where shares are treated as separate items and buyers can submit different bids for every
single share.

\begin{thmappsec}\label{thm:lb_rand_ties_shares}
With a randomized tie-breaking rule, there are some simultaneous first price and second price auction games
for which the Liquid Price of Anarchy is $\Omega\left(\frac{n}{\ncopy}\right)$ when the number of shares is $\ncopy$, even
when agents play pure strategies.
\end{thmappsec}
\begin{proof}
We proceed in a manner similar to the proof of Theorem~\ref{thm:lb_rand_ties}. In particular, we again consider a simultaneous
first price auction in which there are $n$ items and $n$ bidders, where each of the items has $\ncopy$ shares.  Each bidder has the same
valuation profile for the items, namely they all value the first item at $n^4$, and hence extract a value of $\frac{n^4}{\ncopy}$
per share, while the remaining $n-1$ items are valued at $1$, each share of which has value $\frac{1}{\ncopy}$.  Moreover,
each bidder has a budget of $1$ as before.  Observe that the optimal solution is $n$, which is attained by giving each agent
every share of a particular item, for a total Liquid Welfare of $n$ (since each player is budget-capped at $1$).

On the other hand, consider the following randomized tie-breaking rule, where all ties among the highest bids that
occur for each share of item $1$ are broken uniformly at random (the tie-breaking rule for the remaining shares can be arbitrary).
Since item $1$ is valued so highly, the pure strategy profile in which each bidder wants $1$ share of the first item and bids $1$
per share (i.e., their maximum possible bid) is a pure Nash equilibrium.  In particular, each player wins one share of item 1
with probability $\frac{\ncopy}{n}$, and since all $n$ agents are tied for the maximum bid, their utility is
$\frac{\ncopy}{n} \cdot (\frac{n^4}{\ncopy} - 1)$.  If they switch, their utility will be at most $n-1$ (the same holds for second price auctions).
Notice that no player would ever bid on multiple shares of the first item, since this would mean they are bidding strictly less than
one per share and never win a share. This leads to a Liquid Welfare of at most $\ncopy$, which implies the Liquid Price of Anarchy
is at least $\frac{n}{\ncopy}$.
\end{proof}

In fact, we can obtain a similar negative result when players can mix their strategies, even when the tie-breaking rule
is deterministic.  As mentioned earlier, the proof is written in the context of the market clearing mechanism
for shares, but can easily be extended to the case where shares are treated as separate items and players
can submit different bids for each share.

\begin{thmappsec}[Theorem~\ref{thm:lb_mixed_shares}]
With a deterministic tie-breaking rule, there are some simultaneous first price and second price auction games
for which the Liquid Price of Anarchy is $\Omega\left(\frac{n}{\ncopy}\right)$ when the number of shares is $\ncopy$ and agents
mix their strategies.
\end{thmappsec}
\begin{proof}
We proceed in a manner similar to the proof of Theorem~\ref{thm:lb_mixed}, but adapt the example for the share model.
Consider a simultaneous first price auction in which there are $n$ items, each with $\ncopy$ shares, and $n$ bidders,
each with a budget of $1$ (the proof goes through for simultaneous second price auctions as well).
Items $1$ and $2$ have a tie-breaking rule which favors players lexicographically (i.e., player $1$ is the most preferred while
player $n$ is the least preferred), while the rest of the items have an arbitrary tie-breaking rule.  All players value
the first two items at $2^{2n}$, each share of which is valued at $\frac{2^{2n}}{\ncopy}$.
The first $n-2$ players value the remaining items $3,\ldots,n$ at $1$, each share of which is valued at $\frac{1}{\ncopy}$, while players $n-1$
and $n$ value items $3,\ldots,n$ at $0$.  Observe that the optimal solution is $n$, which is attained by allocating item $1$
to player $n-1$, item $2$ to player $n$, and giving each of the remaining $n-2$ items to each of the remaining $n-2$ players
(here, giving an item to a player means giving all shares of the item to the player).

Consider the following bad mixed Nash equilibrium.  Players $n,n-2,\ldots,n-2 \cdot \ncopy$ have pure strategies where all $\ncopy+1$
of them bid their full budget on one share of item $1$ (recall all players have a budget of $1$).  Players $n-1, n-3,\ldots, n-2\cdot\ncopy - 1$
have pure strategies where all $\ncopy + 1$ of them bid their full budget on one share of item $2$.  The first $n-2\cdot\ncopy - 2$ players
randomize their strategies by bidding their full budget on one share of item $1$ with probability $\frac{1}{2}$ and bidding their full budget on
one share of item $2$ with probability $\frac{1}{2}$.  This leads to a Liquid Welfare of $2 \cdot \ncopy$, since only $2 \cdot \ncopy$
players win a share.  Hence, the Liquid Price of Anarchy is $\Omega\left(\frac{n}{\ncopy}\right)$.

Let us see why this mixed strategy profile forms an equilibrium.  Observe that, no matter how the first $n-2\cdot\ncopy - 2$
players randomly choose which item to bid on, there are always at least $\ncopy+1$ players bidding on item $1$ and at least $\ncopy+1$
players bidding on item $2$.  Hence,
player $n$ never wins a share since the tie-breaking rule prefers all other $\ncopy$ agents that purely bid on item $1$,
and similarly player $n-1$ never wins a share since the tie-breaking rule prefers all other $\ncopy$ agents playing purely
on item $2$.  These two players cannot deviate and improve their utility, since they value items $3,\ldots,n$ at $0$ and hence
cannot bid on such items.  Moreover, even if player $n$ deviates and bids their full budget on one share of item $2$,
they still lose due to the tie-breaking rule (similarly, player $n-1$ cannot improve their utility by deviating to item $1$
due to the tie-breaking rule).  In addition, player $n$ is indifferent between bidding $1$ on item $1$ and bidding less than $1$,
since their utility is $0$ either way (similarly, player $n-1$ is indifferent between bidding $1$ on item $2$ and bidding less than $1$).
In fact, every player who bids less than $1$ on item $1$ or less than one on item $2$ loses the item with certainty.

Consider any other player $n - 2 \cdot i$ purely bidding on item $1$ (so $1 \leq i \leq \ncopy$).
There are $\ncopy - i$ players who also purely bid on item $1$ and are preferred by the tie-breaking rule, but if player $n-2 \cdot i$
deviates and purely bids their full budget on one share of item $2$, the probability of winning a share decreases since there are
$\ncopy - i + 1$ players who purely bid on item $2$ and are preferred by the tie-breaking rule.  If player
$n-2 \cdot i$ bids less than $1$ on item $1$ and less than $1$ on item $2$, then their utility is at most $n-2$ (attained by winning
all shares of items $3,\ldots,n$).  However, in the Nash equilibrium, the probability that player $n-2 \cdot i$ wins a share is at least
the probability that all players who mix their strategies randomly choose item $2$, which is given by $\frac{1}{2^{n - 2\ncopy - 2}}$.
Hence, their expected utility is at least $\frac{1}{2^{n - 2\ncopy - 2}}\left(\frac{2^{2n}}{\ncopy} - 1\right) \geq 2^n \geq n-2$, so they cannot improve
their expected utility by deviating.  The argument that players purely bidding on item $2$ cannot improve their expected utility
by deviating is similar.  The only difference is that, for any such player $n - 2\cdot i - 1$ for $1 \leq i \leq \ncopy$,
the number of players that purely bid on item $2$ and are preferred by the tie-breaking rule is $\ncopy-i$, which also holds if
player $n - 2\cdot i - 1$ were to purely bid on item $1$.  Hence, such players would win a share of item $1$ after deviating
with the same probability of winning a share in the mixed Nash equilibrium, so they are indifferent.
Finally, observe that each player $i \geq n-2\ncopy - 2$ is indifferent between item $1$ and item $2$, and they have an expected
utility at least that of any player who bids purely and wins a share with positive probability, which is at least $2^n$.  This is at least
as much as their expected utility in any deviation, and hence we have a mixed Nash equilibrium.
\end{proof}

\section{Bayesian Nash Equilibrium}
\label{sec:appendix_bayesian}
%We consider a share model in which each item $j$ is divided into $\ncopy\ge n$ identical shares, where $n$ is the number of players. For the sake of the analysis we treat each share as a separate item, so that buyers could submit different bids on every single share. A more realistic market clearing mechanism for one item would be the one, where
%\begin{enumerate}
%\item each buyer specifies how many shares she wants to buy and which price she is willing to
      %pay per share
%\item in the decreasing order of bids and until the stock of shares lasts, each buyer receives
      %her demanded number of shares while paying her bid per every purchase.
%\end{enumerate}
%We note that our analysis carries over for this ``clearing house'' item auction with small adjustments which we discus in Section~\ref{sec:extensions}.
We assume agents have additive valuations and submit bids on shares, and if they receive
an $\alloc_{ij}$ fraction of shares of item $j$, then their value is given by $\vali(\alloc_{ij}) = \vij{}\cdot \alloc_{ij}$
($\vali$ is technically defined over sets, but for singleton sets we will often write $\vij{}$ for ease of notation).
We consider a {\em Bayesian} setting, where the bidders' valuations are drawn independently from distributions $\disti[1], \ldots, \disti[n]$
and write $\dists = \disti[1] \times \cdots \times \disti[n]$, so that $\vals$ is drawn from $\dists$.   We assume
that each valuation $\vali$ is realized together with an associated budget $\budgeti$ which we usually will omit in our notations for brevity.
In fact, we will use the notation $\budgeti(\vali)$ to emphasize that each player's budget can be correlated with their
valuation.  We think of $\dists$ as being public knowledge, whereas the realization $\vali$ is known only to agent $i$.

We further assume that the buyers bid according to a Bayesian Nash equilibrium $\bids\sim\strats(\vals)$, where $\vals\sim\dists$.
When buyers bid in simultaneous auctions, this induces a distribution of prices over all shares of items
$\prices\sim\pricedist$ from a distribution $\pricedist$ (e.g., winning bids in first price auctions,
namely $\pricejl[j]{\ell} = \max_i \bid[i]{j}^\ell$ where $\bid[i]{j}^\ell$ is player $i$'s bid for share $\ell$ of item $j$).
In particular, for all items we
can define an ``expected price per item'' at equilibrium or just a ``price per item'' as
$\rprices=(\rpricei[1],\dots,\rpricei[m])$, where $\rpricej=\alpha\sum_{\ell=1}^{\ncopy} \Ex[\pricedists]{\pricejl[j]{\ell}}$, for some $\alpha>1$
($\alpha=2$ will be sufficient for us). This induces a natural ``expected price per share,'' namely $\frac{\rpricej}{\ncopy}$.
One simple observation about $\rprices$ is the following:
\begin{observation}
Revenue is related to prices: $\revenue(\strats)=\frac{1}{\alpha}\sum_{j=1}^{m}\rpricej$, where $\revenue(\strats)$ denotes the expected
revenue at the equilibrium profile $\strats$.
\end{observation}
We next show that if players bid on some fraction of shares of item $j$ uniformly at random according to $\rpricej$,
then they win a large number of shares in expectation.
\begin{claim}\label{cl:copies_bayesian}
For any item $j$, if a player bids on a $\delta$-fraction of shares chosen uniformly at random of item $j$ at a given price
$\frac{\rpricej}{\ncopy}$ per share, then the player receives in expectation at least $\ncopy\cdot\delta\cdot\left(1-\frac{1}{\alpha}\right)$ shares of the item
(i.e., at least a $\delta\cdot\left(1-\frac{1}{\alpha}\right)$-fraction of item $j$).
\end{claim}
\begin{proof}
%Suppose a player bids on a $\delta$-fraction of shares uniformly at random of item $j$.  Since the player bids $\rpricej$ for the item
%and this amount is spread evenly among the shares, this implies the bid per share of item $j$ is $\frac{\rpricej}{\ncopy}$.
%We note that in the specified bidding strategy the expected bid per share of item $j$ is $\frac{\rpricej}{\ncopy}$.
%Let $X_j$ be the random variable which counts the number of shares the player wins.
Suppose towards a contradiction that the expected number of shares won by bidder $i$ is less than
$\delta \cdot \ncopy \cdot \left(1 - \frac{1}{\alpha}\right)$. In particular, it means that
\[\delta \cdot h \cdot \left(1 - \frac{1}{\alpha}\right)>
\sum_{\ell=1}^{\ncopy}\Prx{i\text{ bids on share }\ell}\cdot \Prx[\pricedists]{\pricejl[j]{\ell} < \frac{\rpricej}{\ncopy}}\ge
\sum_{\ell=1}^{\ncopy}\delta \cdot \Prx[\pricedists]{\pricejl[j]{\ell} < \frac{\rpricej}{\ncopy}}
.\]

%Suppose towards a contradiction that $\Ex{X_j} < \delta \cdot h \cdot \left(1 - \frac{1}{\alpha}\right)$.
We further use the definition of $\rpricej$ and Markov's inequality to obtain a contradiction as follows:
\begin{align*}
\frac{\rpricej}{\alpha} = \sum_{\ell=1}^\ncopy\Ex[\pricedists]{\pricejl[j]{\ell}} &\ge
\sum_{\ell=1}^{\ncopy}\frac{\rpricej}{\ncopy} \cdot \Prx[\pricedists]{\pricejl[j]{\ell} \ge \frac{\rpricej}{\ncopy}}\\
&=\sum_{\ell=1}^{\ncopy}\frac{\rpricej}{\ncopy}\left(1 - \Prx[\pricedists]{\pricejl[j]{\ell} < \frac{\rpricej}{\ncopy}}\right) >
\rpricej - \frac{\rpricej}{\ncopy} \cdot \ncopy \cdot\left(1 - \frac{1}{\alpha}\right).
\end{align*}
%
%
%We can write $X_j$ as a sum of indicator random variables $Y_{j}^\ell$, each of which is $1$ if and only if the player receives share $\ell$
%of item $j$.  Hence, $\Ex{X_j} = \sum_{\ell=1}^\ncopy \Ex{Y_j^\ell} \geq \sum_{\ell=1}^\ncopy\delta \cdot \Prx{\pricejl[j]{\ell} < \frac{\rpricej}{\ncopy}}$.  To see why,
%observe that $Y_j^\ell = 1$ if and only if the player randomly chooses share $\ell$, which happens with probability $\delta$, and if the
%bid on share $\ell$, namely $\frac{\rpricej}{\ncopy}$, exceeds the price (i.e., $\pricejl[j]{\ell}$).  Hence, we obtain a contradiction via the following chain of inequalities:
%\begin{align*}
%\frac{\rpricej}{\alpha} = \sum_{\ell=1}^\ncopy\Ex{\pricejl[j]{\ell}} &\geq
%\sum_{\ell=1}^\ncopy\Ex{\pricejl[j]{\ell} | \pricejl[j]{\ell} \geq \frac{\rpricej}{\ncopy}}\Prx{\pricejl[j]{\ell} \geq \frac{\rpricej}{\ncopy}}\\
%&\geq \sum_{\ell=1}^\ncopy\frac{\rpricej}{\ncopy}\Pr\left[\pricejl[j]{\ell} \geq \frac{\rpricej}{\ncopy}\right] =
%\sum_{\ell=1}^\ncopy\frac{\rpricej}{\ncopy}\left(1 - \Pr\left[\pricejl[j]{\ell} < \frac{\rpricej}{\ncopy}\right]\right) >
%\rpricej - \frac{\rpricej}{\ncopy} \cdot \ncopy \cdot (1 - \frac{1}{\alpha}).
%\end{align*}
\end{proof}
When relating prices to Liquid Welfare we notice that
\begin{observation}
Revenue is bounded by the Liquid Welfare:
$\revenue(\strats)\le \LW(\strats)$, where $\LW(\strats)$ denotes the expected Liquid Welfare at the equilibrium profile $\strats$.
\end{observation}
For each fixed profile $\vals$ of additive valuations we consider the following fractional relaxation (LP) of the allocation problem with the goal of optimizing Liquid Welfare:
\begin{alignat*}{3}
 & \text{Maximize} & \sum_{i=1}^n\sum_{j=1}^{m} \vij{}\varalloc_{ij}& \\
 & \text{Subject to} \quad& \sum_{j}\vij{}\varalloc_{ij}& \leq \budgeti(\vali) & \quad \forall & i\\
                 && \sum_{i}\varalloc_{ij} &\leq 1 \quad & \quad \forall & j\\
                 && \varalloc_{ij} &\geq 0 & \quad \forall & i,j
\end{alignat*}
We denote by $\oallocs(\vals)=(\oalloc_{ij})$ the optimal solution to the above LP.
Notice that the optimal fractional solution for the Liquid Welfare would never benefit from
allocating a set of items to a player such that their value for the set exceeds their budget.

We slightly abuse notations by defining $\oalloc_{ij}(\vali)\eqdef\Ex[\valsmi\sim\distsmi]{\oalloc_{ij}(\vali,\valsmi)}$
for each bidder $i$ with fixed valuation $\vali$. We observe that the vector of allocations $\oalloci(\vali)$ is budget feasible, i.e., 
$\sum_{j}\vij{}\oalloc_{ij}(\vali) \le \budgeti(\vali)$, as it is an average of
budget feasible allocations for the fixed valuation $\vali$ and budget $\budgeti(\vali)$.
Furthermore, solutions to this LP give an upper bound on the optimal Liquid Welfare $\opt$.
\begin{observation}
The optimal fractional solution to LP is better than the optimal allocation:
$\sum\limits_{i=1}^n\Ex[\vali\sim\disti]{\sum\limits_{j=1}^{m} \vij{}\cdot\oalloc_{ij}(\vali)}\ge \opt$.
\end{observation}
We now define some notation that will be useful in order to obtain our result.  We
let $q_{ij}(\vals)$ and $q_{ij}(\vali)$ be the expected fraction of shares that player $i$ receives from item $j$ at
an equilibrium strategy $\strats$ for a fixed valuation profile $\vals$ and $\valsmi\sim\distsmi$, respectively.
In addition, for each agent $i$, we consider a set of high value items $J_i(\vali) \eqdef \left\{j \mid \vij{} \ge \rpricej\right\}$.
We further define $Q_i(\vals)$ to be the probability that $\vali(\alloci) \ge \budgeti(\vali)$ at equilibrium for a fixed valuation profile
$\vals$ (recall that $\alloci$ denotes the random set that player $i$ receives in the Bayesian Nash equilibrium).
We similarly define $Q_i(\vali) = \Ex[\valsmi]{Q_i(\vals)}$ for a fixed valuation $\vali$.
We also define three sets of bidders, the first two of which are for budget feasibility reasons and the last of which is for bidders that often
fall under their budget in equilibrium (these sets need not be disjoint).  In particular, for a fixed parameter $\gamma >1$
to be determined later, we define sets $\good_1$, $\good_2$, and $\good_3$:
\begin{align*}
\good_1(\vals)\eqdef\left . \Big\{i~~\middle| \right .  \gamma\sum_{j\in J_i(\vali)} \rpricej\cdot q_{ij}(\vali)&\le\budgeti(\vali)\Big\}, \quad
\good_2(\vals)\eqdef\left . \Big\{ i~~\middle|\right . \sum_{j\in[m]}\frac{\rpricej}{\ncopy}\le\budgeti(\vali)\Big\}, \quad\text{and}\quad \\
\good_3(\vals)&\eqdef\left . \Big\{ i~~\middle|\right . Q_i(\vali) \leq \frac{1}{2 \gamma}\Big\}.
\end{align*}
Throughout our proof, we focus on bidders in the set $\good(\vals) \eqdef \good_1(\vals) \cap \good_2(\vals) \cap \good_3(\vals)$.
We define sets $\notgood_1(\vals)\eqdef[n]\setminus\good_1(\vals)$,
$\notgood_2(\vals)\eqdef[n]\setminus\good_2(\vals)$, $\notgood_3(\vals)\eqdef[n]\setminus\good_3(\vals)$, and $\notgood(\vals)\eqdef[n]\setminus\good(\vals)$.
We note that although the set $\good$ depends on the entire valuation profile $\vals$,  
whether $i\in\good(\vals)$ or $i\in\notgood(\vals)$ is determined only by $\vali$ alone.
To this end, we need
to argue that bidders outside of the set $\good$ do not contribute a lot to the Liquid Welfare at equilibrium $\strats$.
\begin{claim}
\label{cl:notgood_budgets_bayesian}
The total budget of players in $\notgood$ is small:
$\Ex[\vals]{\sum_{i \in \notgood(\vals)}\budgeti(\vali)} < \alpha \cdot \left(\gamma+\frac{n}{\ncopy}\right) \cdot \revenue(\strats) + \Ex[\vals]{\sum_{i \in \notgood_3(\vals)}\budgeti(\vali)}$.
\end{claim}
\begin{proof}
We first consider agents in $\notgood_1$ and obtain
\begin{align*}
\Ex[\vals]{
\sum_{i \in \notgood_1(\vals)}\budgeti(\vali)} &< 
\Ex[\vals]{\gamma\sum_{i \in \notgood_1(\vals)}\sum_{j\in J_i(\vali)} \rpricej\cdot q_{ij}(\vali)} 
\le 
\Ex[\vals]{\gamma\sum_i\sum_{j \in J_i(\vali)}\rpricej\cdot q_{ij}(\vali)} \\
&\le \Ex[\vals]{\gamma \sum_j \rpricej\sum_{i:j \in J_i(\vali)}q_{ij}(\vals)} 
\le \gamma \sum_j \rpricej \leq \gamma \cdot \alpha \cdot \revenue(\strats),
\end{align*}
where the second to last inequality follows from the fact that $\sum_i q_{ij}(\vals) \le 1$ since we have a valid fractional allocation.
Next, we consider agents in $\notgood_2$ and obtain
$$
\Ex[\vals]{\sum_{i \in \notgood_2(\vals)}\budgeti(\vali)} < \Ex[\vals]{\sum_{i \in \notgood_2(\vals)}\sum_{j}\frac{\rpricej}{\ncopy}} 
\le \frac{n}{\ncopy}\sum_j\rpricej \leq \frac{n}{\ncopy} \cdot \alpha \cdot \revenue(\strats).
$$
Combining these bounds for agents in $\notgood_1$ and $\notgood_2$ we have
\begin{align*}
\Ex[\vals]{\sum_{i \in \notgood(\vals)}\budgeti(\vali)} &\le
\Ex[\vals]{\sum_{i \in \notgood_1(\vals)}\budgeti(\vali) + \sum_{i \in \notgood_2(\vals)}\budgeti(\vali) + \sum_{i \in \notgood_3(\vals)}\budgeti(\vali)}\\
&\le \alpha \cdot \left(\gamma+\frac{n}{\ncopy}\right) \cdot \revenue(\strats) + \Ex[\vals]{\sum_{i \in \notgood_3(\vals)}\budgeti(\vali)}.
\end{align*}
\end{proof}
To achieve our result, we consider two main ideas for player deviations in set $\good(\vals)$ (each idea
actually consists of two parts). The first idea is to use the fractional solution to the LP as guidance to claim that players can extract a large amount of
value relative to the optimal solution.  However, for players to deviate, we must first round the fractional solution $\oallocs$
into an integral solution (here, by integral, we mean a multiple of $\frac{1}{\ncopy}$ since this represents the fraction
of shares a player receives out of $\ncopy$ copies).  Define the first LP deviation (integral part) to be $\bids_1^{\infloor{i}}=(\pbidi,\bidsmi)$,
where in $\pbidi$ buyer $i$ bids on a random $\intapprox{\oalloc_{ij}}$-fraction of each item $j\in J_i(\vali)$ with price $\rpricej$
(here, $\intapprox{\oalloc}=\frac{\lfloor\oalloc\cdot\ncopy\rfloor}{\ncopy}$).  Define the second LP deviation (fractional part) to be
$\bids_1^{\{i\}}=(\pbidi,\bidsmi)$, where in $\pbidi$ buyer $i$ bids on a random $\fracapprox{\oalloc_{ij}}$-fraction
of each item $j\in J_i(\vali)$ with price $\rpricej$ (here, $\fracapprox{\oalloc}=\frac{1}{\ncopy}$ if $\oalloc>0$, and
$\fracapprox{\oalloc}=0$ otherwise).  We note that both LP deviations $\bids_1^{\infloor{i}}$ and $\bids_1^{\{i\}}$ are feasible,
since $\vij{} \geq \rpricej$ for every $j \in J_i(\vali)$, and $\sum_{j}\vij{}\cdot\oalloc_{ij} \le \budgeti(\vali)$ as $\oallocs$ is a solution to LP
along with $\good(\vals) \subseteq \good_2(\vals)$.  Moreover, for any $\oalloc_{ij}(\vals)$,
we have $\intapprox{\oalloc_{ij}(\vals)} + \fracapprox{\oalloc_{ij}(\vals)} \ge \oalloc_{ij}(\vals)$.
\begin{lem}[LP deviations]
\label{lem:first_deviation_bayesian}
Buyers in $\good$ at equilibrium $\strats$ derive large value:
\begin{align*}
\Ex[\vals]{\sum_{i \in \good(\vals)}\sum_{j}\vij{}\cdot q_{ij}(\vali)} &\ge
\left(\frac{1}{2}-\frac{1}{2\alpha}\right)\left(\opt - \alpha\left(1+\gamma+\frac{n}{\ncopy}\right)\revenue(\strats) - \Ex[\vals]{\sum_{i \in \notgood_3(\vals)}\budgeti(\vali)}\right).
\end{align*}
\end{lem}
\begin{proof} For the integral part of the LP deviation, since $\Ex[\bids\sim\strats]{\vali(\bids)}\ge\Ex[\bids\sim\strats]{\utili(\bids)}$ and
$\strats$ is a Bayesian Nash equilibrium, we have:
\begin{align}
\label{eq:nash_first_deviation_integral_bayesian}
\Ex[\vals]{\sum_{i \in \good(\vals)}\sum_{j}\vij{}\cdot q_{ij}(\vali)}
&\ge 
\sum_{i}\Ex[\vali]{\ind[i\in\good]{\vali}\Ex[\bids\sim\strats(\vals)]{\utili(\bids)}}\nonumber\\
& \ge \sum_{i}\Ex[\vali]{\ind[i\in\good]{\vali}\Ex[\bidsmi\sim\stratsmi(\valsmi)]{\utili\left(\bids_1^{\infloor{i}}\right)}}
\nonumber\\
&\ge \sum_{i}\Ex[\vali]{\ind[i\in\good]{\vali}\sum_{j \in J_i(\vali)}\left(1-\frac{1}{\alpha}\right)\cdot\intapprox{\oalloc_{ij}(\vali)} \left(\vij{} - \rpricej\right)},
\end{align}
where to derive the last inequality we use Claim~\ref{cl:copies_bayesian}. Similarly, for the fractional part of the LP deviation we have:
\begin{align}
\label{eq:nash_first_deviation_fractional_bayesian}
\Ex[\vals]{\sum_{i \in \good(\vals)}\sum_{j}\vij{}\cdot q_{ij}(\vali)}
&\ge
\sum_{i}\Ex[\vali]{\ind[i\in\good]{\vali}\Ex[\bids\sim\strats(\vals)]{\utili(\bids)}}\nonumber\\
& \ge \sum_{i}\Ex[\vali]{\ind[i\in\good]{\vali}\Ex[\bidsmi\sim\stratsmi(\valsmi)]{\utili\left(\bids_1^{\infloor{i}}\right)}}\nonumber\\
&\ge \sum_{i}\Ex[\vali]{\ind[i\in\good]{\vali}\sum_{j \in J_i(\vali)}\left(1-\frac{1}{\alpha}\right)\cdot\fracapprox{\oalloc_{ij}(\vali)} \left(\vij{} - \rpricej\right)},
\end{align}
Combining Equation~\eqref{eq:nash_first_deviation_integral_bayesian} and Equation~\eqref{eq:nash_first_deviation_fractional_bayesian} we get
%\begin{align}
%\label{eq:nash_first_deviation}
%2\sum_{i\in\good}\sum_{j}\vij{}\cdot q_{ij} &\ge
%\sum_{i \in\good}\sum_{j \in J_i}\left(1-\frac{1}{\alpha}\right)\cdot\left(\intapprox{\oalloc_{ij}}+\fracapprox{\oalloc_{ij}}\right) \left(\vij{} - \rpricej\right)
%\\
%&\ge
%\sum_{i \in\good}\sum_{j \in J_i}\left(1-\frac{1}{\alpha}\right)\cdot \oalloc_{ij}\left(\vij{} - \rpricej\right)
%\geq
%\left(\frac{1}{2}-\frac{1}{2\alpha}\right)\sum_{i \in\good}\sum_{j \in J_i} \vij{}\cdot\oalloc_{ij}, \nonumber
%\end{align}
%where the last inequality follows from the fact that $\vij{} \geq 2\rpricej$ for each $j \in J_i$.
\begin{align}
\label{eq:nash_first_deviation_bayesian}
2\Ex[\vals]{\sum_{i \in \good(\vals)}\sum_{j}\vij{}\cdot q_{ij}(\vali)} 
&\ge
\sum_{i}\Ex[\vali]{\ind[i\in\good]{\vali}\sum_{j \in J_i(\vali)}\left(1-\frac{1}{\alpha}\right) \cdot\left(\intapprox{\oalloc_{ij}(\vali)}+\fracapprox{\oalloc_{ij}(\vali)}\right)\left(\vij{} - \rpricej\right)}\nonumber\\
%%\sum_{i \in\good}\sum_{j \in J_i}\left(1-\frac{1}{\alpha}\right)\cdot(\left\intapprox{\oalloc_{ij}}+\fracapprox{\oalloc_{ij}}\right) \left(\vij{} - \rpricej\right)
&\ge \sum_{i}\Ex[\vali]{\ind[i\in\good]{\vali}\sum_{j \in J_i(\vali)}\left(1-\frac{1}{\alpha}\right)\oalloc_{ij}(\vali)\left(\vij{} - \rpricej\right)}\nonumber\\
& = \left(1-\frac{1}{\alpha}\right)\left(\Ex[\vals]{\sum_{i \in\good(\vals)}\sum_{j \in J_i(\vali)} \vij{}\cdot\oalloc_{ij}(\vals) -
\sum_{i \in\good(\vals)}\sum_{j \in J_i(\vali)}\rpricej\cdot\oalloc_{ij}(\vals)}\right).
\end{align}
We further estimate
\begin{align*}
\Ex[\vals]{\sum_{i \in \good(\vals)}\sum_{j\in J_i(\vali)}\vij{}\cdot\oalloc_{ij}(\vals)} &=
\Ex[\vals]{\sum_{i,j}\vij{}\cdot\oalloc_{ij}(\vals) - \sum_{i\not\in\good(\vals)}\sum_{j}\vij{}\cdot\oalloc_{ij}(\vals) - \sum_{i\in\good(\vals)}\sum_{j \not\in J_i(\vali)}\vij{}\cdot\oalloc_{ij}(\vals)}\\
&\ge 
\Ex[\vals]{\sum_{i,j}\vij{}\cdot\oalloc_{ij}(\vals) - \sum_{i\not\in\good(\vals)}\budgeti(\vali) - \sum_{i \in \good(\vals)}\sum_{j \not\in J_i(\vali)}\rpricej\cdot \oalloc_{ij}(\vals)},
\end{align*}
where in the last inequality we used the condition from the LP that $\sum_{j}\vij{}\cdot\oalloc_{ij}(\vals)\le\budgeti(\vali)$ and that $\vij{}\le\rpricej$ for each $j\not\in J_i(\vali)$.
We substitute the last estimate into Equation~\eqref{eq:nash_first_deviation_bayesian} and obtain a lower bound on $2\Ex[\vals]{\sum_{i \in \good(\vals)}\sum_{j}\vij{}\cdot q_{ij}(\vali)}$:
%\begin{align*}
%\sum_{i \in \good}\sum_{j\in J_i}\vij{}\cdot\oalloc_{ij} &=
%\sum_{i,j}\vij{}\cdot\oalloc_{ij} - \sum_{i\not\in\good}\sum_{j}\vij{}\cdot\oalloc_{ij} - \sum_{i\in\good}\sum_{j \not\in J_i}\vij{}\cdot\oalloc_{ij}\\
%&\ge \sum_{i,j}\vij{}\cdot\oalloc_{ij} - \sum_{i\not\in\good}\budgeti - 2\sum_{i \in \good}\sum_{j \not\in J_i}\rpricej\cdot \oalloc_{ij},
%\end{align*}
%where in the last inequality we use the condition from the LP that $\sum_{j}\vij{}\cdot\oalloc_{ij}\le\budgeti$ for each $i$ and that $\vij{} \le 2\rpricej$ for each $j\not\in J_i$.
%We substitute the last estimate into Equation~\eqref{eq:nash_first_deviation} and get
\begin{align*}
\left(1-\frac{1}{\alpha}\right) & \Ex[\vals]{
\sum_{i,j}\vij{}\cdot\oalloc_{ij}(\vals) - \sum_{i\not\in\good(\vals)}\budgeti(\vali) - \sum_{i \in \good(\vals)}\sum_{j \not\in J_i(\vali)}\rpricej\cdot \oalloc_{ij}(\vals) -
\sum_{i \in\good(\vals)}\sum_{j \in J_i(\vali)}\rpricej\cdot\oalloc_{ij}(\vals)}\\
&=
\left(1-\frac{1}{\alpha}\right)\Ex[\vals]{
\sum_{i,j}\vij{}\cdot\oalloc_{ij}(\vals) - \sum_{i \in \good(\vals)}\sum_{j}\rpricej\cdot \oalloc_{ij}(\vals) - \sum_{i\not\in\good(\vals)}\budgeti(\vali)}\\
&\ge
\left(1-\frac{1}{\alpha}\right)\left(\opt - \alpha\left(1+\gamma+\frac{n}{\ncopy}\right)\revenue(\strats) - \Ex[\vals]{\sum_{i \in \notgood_3(\vals)}\budgeti(\vali)}\right),
\end{align*}
where the last inequality follows from the LP constraint that $\sum_i \oalloc_{ij}(\vals) \leq 1$ for each $j$,
the observation $\sum_{j}\rpricej= \alpha\cdot\revenue(\strats)$, and Claim~\ref{cl:notgood_budgets_bayesian}.
%\begin{align*}
%2\sum_{i\in\good}\sum_{j}\vij{}\cdot q_{ij} &\ge
%\left(\frac{1}{2}-\frac{1}{2\alpha}\right)\left(
%\sum_{i,j}\vij{}\cdot\oalloc_{ij} - \sum_{i\not\in\good}\budgeti - 2\sum_{i \in \good}\sum_{j \not\in J_i}\rpricej\cdot \oalloc_{ij}\right)\\
%&\geq
%\left(\frac{1}{2}-\frac{1}{2\alpha}\right)\left(
%\sum_{i,j}\vij{}\cdot\oalloc_{ij} - 2\sum_{i \in \good}\sum_{j}\rpricej\cdot \oalloc_{ij} - \sum_{i\not\in\good}\budgeti\right)\\
%&\ge
%\left(\frac{1}{2}-\frac{1}{2\alpha}\right)\left(\opt - \alpha\left(2+\gamma+\frac{n}{\ncopy}\right)\revenue(\strats)\right),
%\end{align*}
%where the last inequality follows from the LP constraint that $\sum_i \oalloc_{ij} \leq 1$ for each $j$, $\sum_{j}\rpricej= \alpha\cdot\revenue(\strats)$,
%and Claim~\ref{cl:notgood_budgets}.
\end{proof}
%%%%%%%%%%%%%%%%%%%%%%%%%%%%%%%%%%%%%%%%%%%%%%%%%%%%%%%%%%%%%%%%%%%%%%%%%%%%%%%%%%%%%%%%%%%%%%%%%%%%%%%%%%%%%%%%%%%%%%%%%%%
We now turn to our second type of deviation, but we need to further restrict the set of items that players bid on.
In particular, we let $\Gamma_{i}(\vali) = \left\{j ~~\middle|~~ q_{ij}(\vali) \leq \frac{1}{\gamma}\right\}$, and
define $G_i(\vali) = J_i(\vali) \cap \Gamma_i(\vali)$.  We now define the $\gamma$-boosting deviation (integral part)
as $\bids_2^{\infloor{i}}=(\pbidi,\bidsmi)$, where in $\pbidi$ buyer $i$ bids on a random
$\intapprox{\gamma\cdot q_{ij}(\vali)}$-fraction of each item $j\in G_i(\vali)$ with price $\rpricej$, where $\gamma > 1$
is a constant to be determined later. Note that each $\bids_2^{\infloor{i}}$ deviation for every $i\in\good(\vals)$ is feasible
since $\good(\vals) \subseteq \good_1(\vals)$. Similarly, we define the fractional part of the $\gamma$-boosting deviation
as $\bids_2^{\{i\}}$, which is also a feasible deviation since $\good(\vals) \subseteq \good_2(\vals)$.  Also, since players
bid on items in $G_i(\vals) \subseteq \Gamma_i(\vals)$, we have $\intapprox{\gamma \cdot q_{ij}(\vali)} \leq 1$ (we also have
$\fracapprox{\gamma \cdot q_{ij}(\vali)} \leq 1$, which holds for all items by definition).
\begin{lem}[$\gamma$-boosting deviation]
The value derived by buyers in $\good$ is comparable to the Liquid Welfare obtained at equilibrium:
\label{lem:second_deviation_bayesian}
\bee
\left(1-\frac{2\alpha}{\gamma(\alpha-1)}\right)\Ex[\vals]{\sum_{i\in\good(\vals)} \sum_j \vij{}\cdot q_{ij}(\vali)}
\le
\alpha \cdot \revenue(\strats) + 2 \cdot \LW(\strats) -
\frac{1}{\gamma}\Ex[\vals]{\sum_{i \in \notgood_3(\vals)}B_i(\vali)}.
\eee
\end{lem}
\begin{proof}
For the integral part of the $\gamma$-boosting deviation, we can now again obtain bounds via the Bayesian Nash equilibrium condition and Claim~\ref{cl:copies_bayesian}:
\begin{align*}
\Ex[\vals]{\sum_{i \in \good(\vals)}\sum_{j}\vij{}\cdot q_{ij}(\vali)}
&\ge
\sum_{i}\Ex[\vali]{\ind[i\in\good]{\vali}\Ex[\bids\sim\strats(\vals)]{\utili(\bids)}}
\\
& \ge \sum_{i}\Ex[\vali]{\ind[i\in\good]{\vali}\Ex[\bidsmi\sim\stratsmi(\valsmi)]{\utili\left(\bids_2^{\infloor{i}}\right)}}
\\
&\ge \sum_{i}\Ex[\vali]{\ind[i\in\good]{\vali}\sum_{j \in G_i(\vali)}\left(1-\frac{1}{\alpha}\right)
\intapprox{\gamma \cdot q_{ij}(\vali)}\left(\vij{} - \rpricej\right)}.
\end{align*}
Similarly, for the fractional part of the $\gamma$-boosting deviation we get:
\begin{align*}
\Ex[\vals]{\sum_{i \in \good(\vals)}\sum_{j}\vij{}\cdot q_{ij}(\vali)}
&\ge \sum_{i}\Ex[\vali]{\ind[i\in\good]{\vali}\Ex[\bidsmi\sim\stratsmi(\valsmi)]{\utili\left(\bids_2^{\infloor{i}}\right)}}
\\
&\ge \sum_{i}\Ex[\vali]{\ind[i\in\good]{\vali}\sum_{j \in G_i(\vali)}\left(1-\frac{1}{\alpha}\right)
\fracapprox{\gamma \cdot q_{ij}(\vali)}\left(\vij{} - \rpricej\right)}.
\end{align*}
Together these two deviations give us
\begin{align}
\label{eq:nash_second_deviation_bayesian}
&2\Ex[\vals]{\sum_{i \in \good(\vals)}\sum_{j}\vij{}\cdot q_{ij}(\vali)}
\ge \sum_{i}\Ex[\vali]{\ind[i\in\good]{\vali}\sum_{j \in G_i(\vali)}\left(1-\frac{1}{\alpha}\right)
\gamma \cdot q_{ij}(\vali)\left(\vij{} - \rpricej\right)}.
\end{align}
%\begin{align}
%\label{eq:nash_second_deviation}
%2\sum_{i\in\good}\sum_{j}\vij{}\cdot q_{ij} &\ge
%\sum_{i \in\good}\sum_{j \in J_i}\left(1-\frac{1}{\alpha}\right) \min\left(\gamma\cdot q_{ij},1\right)\left(\vij{} - \rpricej\right)
%\nonumber\\
%&\ge \left(\frac{1}{2}-\frac{1}{2\alpha}\right)
%\left(\gamma\sum_{i \in\good} \sum_{j \in J_i}\vij{}\cdot q_{ij} -  \sum_{i\in\good}\sum_{j \in J_i}\vij{}\max\left(\gamma\cdot q_{ij} - 1,0\right)\right),
%\end{align}
%where the last inequality follows from the fact that $\vij{} \geq 2\rpricej$ for each $j \in J_i$.
We further estimate the term $\sum_{i}\Ex[\vali]{\ind[i\in\good]{\vali}\sum_{j \in G_i(\vali)}q_{ij}(\vali)\left(\vij{} - \rpricej\right)}$ on the
RHS of Equation~\eqref{eq:nash_second_deviation_bayesian}, which can be rewritten as:
\begin{align}
&\Ex[\vals]{\sum_{i \in\good(\vals)} \sum_{j \in G_i(\vali)}\left(\vij{}-\rpricej\right)\cdot q_{ij}(\vali)}\nonumber \\
%%\sum_{i\in\good} \sum_{j \in J_i}\left(\vij{}-\rpricej\right)\cdot q_{ij}
&=\Ex[\vals]{\sum_{i\in\good(\vals)} \sum_j \vij{}\cdot q_{ij}(\vali) - \sum_{i \in \good(\vals)} \sum_{j \not\in G_i(\vali)} \vij{}\cdot q_{ij}(\vali)-
\sum_{i\in\good(\vals)} \sum_{j \in G_i(\vali)}\rpricej\cdot q_{ij}(\vali)}\nonumber\\
&\ge
\Ex[\vals]{\sum_{i\in\good(\vals)} \sum_j \vij{}\cdot q_{ij}(\vali) - \sum_{i \in \good(\vals)} \sum_{j \not\in \Gamma_i(\vali)}\vij{}\cdot q_{ij}(\vali)-
\sum_{i\in\good(\vals)} \sum_{j}\rpricej\cdot q_{ij}(\vali)}\nonumber \\
&\geq
\Ex[\vals]{\sum_{i\in\good(\vals)} \sum_j \vij{}\cdot q_{ij}(\vali)} -
\Ex[\vals]{\sum_{i\in\good(\vals)} \sum_{j \not\in \Gamma_i(\vali)} \vij{}\cdot q_{ij}(\vali)} - \alpha \cdot \revenue(\strats),
\label{eq:first_term_boosting_bayesian}
\end{align}
where the first inequality holds as $\sum_{j \not\in G_i(\vali)}\vij{}q_{ij}(\vali) \leq
\sum_{j \not\in J_i(\vali)}\vij{}q_{ij}(\vali) + \sum_{j \not\in \Gamma_i(\vali)}\vij{}q_{ij}(\vali)$
and $\vij{} < \rpricej$ for every $j\notin J_i(\vali)$, and the last inequality holds as
$\Ex[\vals]{\sum_i q_{ij}(\vali)} = \Ex[\vals]{\sum_i q_{ij}(\vals)} \leq~1$.  Our next goal will be to bound the term
$\Ex[\vals]{\sum_{i\in\good(\vals)} \sum_{j \not\in \Gamma_i(\vali)} \vij{}\cdot q_{ij}(\vali)}$
on the RHS of Equation~\eqref{eq:first_term_boosting_bayesian}.
% Note that we have
% $\Ex[\vals]{\sum_{i\in\good(\vals)} \sum_j \vij{}\cdot q_{ij}(\vals)} = \Ex[\vals]{\sum_{i\in\good(\vals)} \sum_j \vij{}\cdot q_{ij}(\vali)}$.
Before that we need to do some preparations.  To ease the notations we denote by $\sharej$ the $\ell^{th}$ share of item $j$.
We observe that the expected Liquid Welfare at equilibrium can be written as
$\LW(\strats) = \sum_{i}\Ex[\vals]{\Ex[\bids\sim\strats(\vals)]{\min\{\vali(\alloci),\budgeti(\vali)\}}}$.
We let $\LW(\strats(\vals)) = \sum_{i}\Ex[\bids\sim\strats(\vals)]{\min{\vali(\alloci),\budgeti(\vali)}}$, so that we have
$\LW(\strats) = \Ex[\vals]{\LW(\strats(\vals))}$.  For any fixed valuation profile $\vals$, $\LW(\strats(\vals))$ is given by:
\begin{align*}
& \sum_{i} \prob[\bids\sim\strats(\vals)]{\vali(\alloci)>\budgeti(\vali)}\cdot\budgeti(\vali) 
+\sum_{i,j}\sum_{\ell=1}^{\ncopy}\prob[\bids\sim\strats(\vals)]{\{\vali(\alloci)\le\budgeti(\vali)\}\land\{i \text{ wins }\sharej\}}\cdot\frac{\vij{}}{\ncopy}\\
&=\sum_{i} Q_i(\vals)\cdot\budgeti(\vali)+
\sum_{i,j}\frac{\vij{}}{\ncopy}\left(
\sum_{\ell=1}^{\ncopy}\prob{i \text{ wins }\sharej}-\sum_{\ell=1}^{\ncopy}\prob{\{\vali(\alloci)>\budgeti(\vali)\}\land\{i \text{ wins }\sharej\}}\right)\\
&=\sum_{i} Q_i(\vals)\cdot\budgeti(\vali)+ \sum_{i,j}\vij{}\cdot
\max\left\{0~,~q_{ij}(\vals)-\frac{1}{\ncopy}\sum_{\ell=1}^{\ncopy}\prob{\{\vali(\alloci)>\budgeti(\vali)\}\land\{i \text{ wins }\sharej\}}\right\} \\
&\ge \sum_{i} Q_i(\vals)\cdot\budgeti(\vali)+ \sum_{i,j}\vij{}\cdot\max\left\{0,q_{ij}(\vals)-\prob{\vali(\alloci)>\budgeti(\vali)}\right\} \\
&= \sum_{i} Q_i(\vals)\cdot\budgeti(\vali) +\sum_{i,j}\max\left\{0,q_{ij}(\vals)-Q_i(\vals)\right\}\cdot\vij{},
\end{align*}
where the second equality holds true as the expression inside the $\max$ cannot be negative and $q_{ij}(\vals)=\frac{1}{\ncopy}\sum_{\ell=1}^{\ncopy}\prob{i \text{ wins }\sharej}$
by definition of $q_{ij}(\vals)$, the first inequality holds since
$\prob{\vali(\alloci)>\budgeti(\vali)}\ge\prob{\{\vali(\alloci)>\budgeti(\vali)\}\land\{i \text{ wins }\sharej\}}$, and the last equality holds by definition of $Q_i(\vals)$.
Taking expectation over both sides, we have:
%\begin{align}
%\label{eq:welfare_budgets}
%\LW(\strats) &= \sum_{i} \prob[\bids\sim\strats]{\vali(\alloci)>\budgeti}\cdot\budgeti
%+\sum_{i}\sum_{j}\prob[\bids\sim\strats]{\{\vali(\alloci)\le\budgeti\}\text{ and }\{i \text{ wins }j\}}\cdot\vij{}
%\nonumber\\
%&\ge \sum_{i} \prob[\bids\sim\strats]{\vali(\alloci)>\budgeti}\cdot\budgeti
%+\sum_{i,j}\max\Big(0,\prob{i \text{ wins }j}-\prob{\vali(\alloci)>\budgeti}\Big)\cdot\vij{}\nonumber\\
%&\ge
%\sum_{i} Q_i\cdot\budgeti +\sum_{i,j}\max\Big(0,q_{ij}-Q_i\Big)\cdot\vij{}
%\end{align}
\begin{align}
\label{eq:welfare_budgets_bayesian}
\LW(\strats) &\geq \Ex[\vals]{\sum_{i \in \notgood_3(\vals)} Q_i(\vals)\cdot\budgeti(\vali)} +
\Ex[\vals]{\sum_{i\in\good(\vals)}\sum_{j}\max\left\{0,q_{ij}(\vals)-Q_i(\vals)\right\}\cdot\vij{}} \nonumber\\
&\geq \sum_{i} \Ex[\vali]{\ind[i\in\notgood_3]{\vali}Q_i(\vali) \cdot \budgeti(\vali)} +
\sum_{i} \Ex[\vali]{\ind[i \in \good]{\vali}\sum_{j\not\in\Gamma_i(\vali)}(q_{ij}(\vali) - Q_i(\vali))\cdot\vij{}} \nonumber\\
&\geq \frac{1}{2\gamma}\sum_{i} \Ex[\vali]{\ind[i\in\notgood_3]{\vali}\budgeti(\vali)} +
\frac{1}{2}\sum_{i} \Ex[\vali]{\ind[i \in \good]{\vali}\sum_{j\not\in\Gamma_i(\vali)}\vij{} \cdot q_{ij}(\vali)} \nonumber\\
&= \frac{1}{2\gamma}\Ex[\vals]{\sum_{i\in\notgood_3(\vals)}\budgeti(\vali)} +
\frac{1}{2}\Ex[\vals]{\sum_{i \in \good(\vals)}\sum_{j\not\in\Gamma_i(\vali)}\vij{} \cdot q_{ij}(\vali)},
\end{align}
where to obtain the first inequality we restrict the set of players that we sum over;
the second inequality follows from the facts that $q_{ij}(\vali) = \Ex[\valsmi]{q_{ij}(\vals)}$ and $Q_i(\vali) = \Ex[\valsmi]{Q_i(\vals)}$,
along with the facts that  $\max\{0,q_{ij}(\vals)-Q_i(\vals)\} \geq 0$ (which we apply for $j \in \Gamma_i(\vali)$) and
$\max\left\{0,q_{ij}(\vals)-Q_i(\vals)\right\} \geq q_{ij}(\vals)-Q_i(\vals)$ (which we apply for $j \not\in \Gamma_i(\vali)$);
the third inequality holds since players $i \in \notgood_3(\vals)$ have $Q_i(\vali) > \frac{1}{2\gamma}$, while for players
$i \in \good(\vals) \subseteq \good_3(\vals)$ and items $j \not\in \Gamma_i(\vali)$
we have $Q_i(\vali) \leq \frac{1}{2} \cdot \frac{1}{\gamma} \leq \frac{q_{ij}(\vali)}{2}$.
% \begin{align}
% \sum_{i\in\good(\vals)} & \sum_{j \in J_i(\vali)}\vij{}\cdot\max\left(\gamma\cdot q_{ij}(\vals) - 1,0\right)\nonumber\\
% &= \sum_{i\in\good(\vals)\cap T(\vals)}\sum_{j \in J_i(\vali)}\vij{}\max\left(\gamma\cdot q_{ij}(\vals) - 1,0\right)
% +\sum_{i\in\good(\vals)\setminus T(\vals)}\sum_{j \in J_i(\vali)}\vij{}\cdot\max\left(\gamma\cdot q_{ij}(\vals) - 1,0\right)\nonumber\\
% &\le \sum_{i\in\good(\vals)\cap T(\vals)}\sum_{j \in J_i(\vali)}\vij{}\cdot\gamma\cdot q_{ij}(\vals)
% +\gamma\sum_{i\in\good(\vals)\setminus T(\vals)}\sum_{j \in J_i(\vali)}\vij{}\cdot\max\left(q_{ij}(\vals) - \frac{1}{\gamma},0\right)\nonumber\\
% &\le\sum_{i\in\good(\vals)\cap T(\vals)}\budgeti(\vali)+\gamma\sum_{i\not\in T(\vals)}\sum_{j}\vij{}\cdot\max\left(q_{ij}(\vals) - Q_i(\vals),0\right)
% \le \gamma\cdot\LW(\strats),
% \label{eq:second_term_boosting_bayesian}
% \end{align}
% where to obtain the first inequality we substituted $\max\left(\gamma\cdot q_{ij}(\vals) - 1,0\right)$ with $\gamma\cdot q_{ij}(\vals)$;
% the second inequality follows from the facts that $\gamma\sum_{j \in J_i(\vali)}\vij{}\cdot q_{ij}(\vals)\le\budgeti(\vali)$ for any $i\in\good(\vals)$
% and that $Q_i(\vals)<\frac{1}{\gamma}$ for any $i\notin T(\vals)$; the last inequality is simply \eqref{eq:welfare_budgets_bayesian}.
Now we rearrange terms from Equation~\eqref{eq:welfare_budgets_bayesian} to get
$\Ex[\vals]{\sum_{i \in \good(\vals)}\sum_{j \not\in \Gamma_i(\vali)}\vij{}q_{ij}(\vali)} \leq
2\cdot\LW(\strats) - \frac{1}{\gamma}\Ex[\vals]{\sum_{i \in \notgood_3(\vals)}\budgeti(\vali)}$.
Combining Equation~\eqref{eq:nash_second_deviation_bayesian} and Equation~\eqref{eq:first_term_boosting_bayesian}, we can substitute this
upper bound to get:
\begin{align*}
2 & \Ex[\vals]{\sum_{i\in\good(\vals)}\sum_{j}\vij{}\cdot q_{ij}(\vali)}\\
&\ge \left(1-\frac{1}{\alpha}\right)\gamma
\left(\Ex[\vals]{\sum_{i\in\good(\vals)} \sum_j \vij{}\cdot q_{ij}(\vali)} - \alpha \cdot \revenue(\strats) -
2\cdot\LW(\strats) + \frac{1}{\gamma}\Ex[\vals]{\sum_{i \in \notgood_3(\vals)}\budgeti(\vali)}\right).
\end{align*}
Dividing both sides by $\left(1-\frac{1}{\alpha}\right)\gamma$ and rearranging terms gives the lemma:
\bee
\left(1-\frac{2\alpha}{\gamma(\alpha-1)}\right)\Ex[\vals]{\sum_{i\in\good(\vals)} \sum_j \vij{}\cdot q_{ij}(\vali)}
\le
\alpha \cdot \revenue(\strats) + 2 \cdot \LW(\strats) -
\frac{1}{\gamma}\Ex[\vals]{\sum_{i \in \notgood_3(\vals)}B_i(\vali)}.
\eee
\end{proof}
Finally, we show that the Liquid Price of Anarchy of any Bayesian Nash equilibrium is bounded.
\begin{thmapp}[Theorem~\ref{thm:mixed_first_price_bayesian}]
In simultaneous first-price auctions with $n$ additive bidders and budgets where every item
has $\ncopy$ equal shares (copies), the Liquid Price of Anarchy of Bayesian Nash equilibria is
$O\left(1+\frac{n}{\ncopy}\right)$ (at most $51.5$, when $\ncopy\ge n$).
\end{thmapp}
\begin{proof}
We combine the bounds from Lemma~\ref{lem:first_deviation_bayesian} and Lemma~\ref{lem:second_deviation_bayesian} and obtain
\begin{align*}
\alpha \cdot \revenue(\strats) &+ 2 \cdot \LW(\strats) - \frac{1}{\gamma} \Ex[\vals]{\sum_{i \in \notgood_3(\vals)}\budgeti(\vali)}
\ge \\
&\left(1-\frac{2\alpha}{\gamma(\alpha-1)}\right)\left(\frac{1}{2}-\frac{1}{2\alpha}\right)
\left(\opt - \alpha\left(1+\gamma+\frac{n}{\ncopy}\right)\revenue(\strats) - \Ex[\vals]{\sum_{i \in \notgood_3(\vals)}\budgeti(\vali)} \right).
\end{align*}
Since $\LW(\strats)\ge\revenue(\strats)$ we further derive that
\begin{align*}
\left(\alpha + 2 + \frac{1}{2}\left(1-\frac{1}{\alpha}-\frac{2}{\gamma}\right)\alpha\left(1+\gamma+\frac{n}{\ncopy}\right)\right) \LW(\strats) &\geq \\
\frac{1}{2}\left(1-\frac{1}{\alpha}-\frac{2}{\gamma}\right)\opt &+
\left(\frac{1}{\gamma} - \frac{1}{2}\left(1-\frac{1}{\alpha}-\frac{2}{\gamma}\right)\right)\Ex[\vals]{\sum_{i \in \notgood_3(\vals)}\budgeti(\vali)}.
\end{align*}
As long as the factor in front of $\Ex[\vals]{\sum_{i \in \notgood_3(\vals)}\budgeti(\vali)}$ is nonnegative,
we have $\opt\le O\left(\frac{n}{\ncopy}\right)\cdot\LW(\strats)$ for any $1 \leq \ncopy \leq n$ for a particular choice of
parameters (e.g., $\alpha = 2.26,\gamma=7.16$).  In particular, when $\ncopy \geq n$, we have that the $\lpoa$ is at most $51.5$.
%Buyers in $\good$ at equilibrium $\strats$ derive large value:
%\begin{align*}
%\sum_{i \in \good}\sum_{j}\vij{}\cdot q_{ij} &\ge \left(\frac{1}{2}-\frac{1}{2\alpha}\right)\left(\opt - \alpha\left(1+\gamma+\frac{n}{\ncopy}\right)\revenue(\strats)\right).
%\end{align*}
\end{proof}

\section{Tightness Results for Simple Auctions}
\label{sec:pureExamples}

In this section, we show that Theorem~\ref{thm:1paupper_app} and Theorem~\ref{thm:2paupper} are essentially tight by
giving an explicit game for which the Liquid Price of Anarchy of both first price and second price auctions is arbitrarily
close to $2$.  In fact, agents in our lower bound only have additive valuation functions.

\begin{thmapp}
There is a simultaneous second price auction game and a simultaneous first price auction game which have a Nash equilibrium $\bids$ such
that the Liquid Price of Anarchy is arbitrarily close to $2$.
\end{thmapp}

\begin{proof}
Consider a simultaneous second price auction game with $m=2$ items and $n=2$ players, and fix any $\epsilon > 0$.  Player $1$ has a budget of
$B_1=10-\epsilon$, and a value of $10$ for item $1$ and a value of $0$ for item $2$.  Player $2$ has a budget of $B_2=10$, and a value of $10$ for both items.
The player valuations are additive, so their value for a bundle is simply the sum of the values of each item.
The optimal solution splits the two items between the two players, giving item $1$ to player $1$ and item $2$ to player $2$.
The Liquid Welfare of this solution is $\opt = \min\{v_1(\{1\}),B_1\} + \min\{v_2(\{2\}),B_2\} = 10-\epsilon + 10 = 20 - \epsilon$.

On the other hand, there exists the following Nash equilibrium $\bids$.  Suppose player $1$ bids $0$ for both items, while player $2$
bids $10-\frac{\epsilon}{2}$ for item $1$ and $\frac{\epsilon}{2}$ for item $2$.  For these bids, player $2$ wins both items, which results in
a Liquid Welfare of $\LW(\bids) = \min\{v_2(\{1,2\}),B_2\} = 10$.  To see why $\bids$ is a Nash equilibrium, observe that player
$1$ is not interested in bidding for item $2$ since their value is $0$ for the item.  Moreover, player $1$'s budget is not high enough to
outbid player $2$ for item $1$, and hence player $1$'s utility cannot be improved (even though it is $0$).  Player $2$'s utility is as high
as it can be, since player $2$ gets both items and pays $0$ for a utility of $20$.  Hence, the Liquid Price of Anarchy is given by
$\frac{\opt}{\LW(\bids)} = \frac{20-\epsilon}{10}$.

The same setup shows that the Liquid Price of Anarchy of simultaneous first price auctions is arbitrarily close to $2$.
In fact, for the same game, essentially the same Nash equilibrium $\bids$ exists for the first price auction setting.  The only difference
is that player $2$ can improve their utility by bidding less for the two items.  If we assume that players' bids must be multiples
of some fixed value, then player $2$ must bid slightly above $0$ for item $1$ and slightly above $10-\epsilon$ for item $2$,
and now player $2$ cannot improve their utility by bidding less, since doing so may result in losing one or more items.
\end{proof}

\end{document}